\documentclass[twocolumn,10pt,superscriptaddress,amsmath,amssymb,aps,pra,floatfix,]{revtex4-2}

\usepackage{graphicx}% Include figure files
\usepackage{bm}% bold math
\usepackage{times}
\usepackage{amsmath,amssymb}
\usepackage{amsthm}
\usepackage{enumerate}
\usepackage{multirow}
\usepackage{siunitx}
\usepackage{makecell}
\usepackage{physics}
\usepackage{longtable,array}
\usepackage[caption=false]{subfig}
\usepackage{tikz}
\usepackage{xcolor}
\usetikzlibrary{quantikz2}
\usepackage{tabularx}
\usepackage{soul}
\usepackage{xcolor}
\setstcolor{red}

\usepackage[colorlinks=true,linkcolor=blue,citecolor=red, linktocpage=true,breaklinks=true]{hyperref}

\newcommand{\eq}{\begin{equation}}
\newcommand{\en}{\end{equation}}
\newcommand{\eqa}{\begin{eqnarray}}
\newcommand{\ena}{\end{eqnarray}}
\newcommand{\pr}{\mathrm{Pr}}
\newcommand{\E}{\mathrm{E}}

\newtheorem{theorem}{Theorem} % 主计数器
\newtheorem{lemma}[theorem]{Lemma} % 与 theorem 共享计数器
\newtheorem{corollary}[theorem]{Corollary}
\newtheorem{definition}[theorem]{Definition}
\begin{document}

\title{Asymptotic bounds on quantum partial search algorithm and its applications to parallel search}% Force line breaks with \\

\author{Yan-Bo \surname{Jiang}}
\affiliation{School of Physics, Northwest University, Xi’an 710127, China}

\author{Xiao-Hui Wang}
%\email{xhwang@nwu.edu.cn}
\affiliation{School of Physics, Northwest University, Xi’an 710127, China}
\affiliation{Shaanxi Key Laboratory for Theoretical Physics Frontiers, Xi'an 710127, China}
\affiliation{Peng Huanwu Center for Fundamental Theory, Xi'an 710127, China}
\affiliation{Fundamental Discipline Research Center for Quantum Science and Technology of Shaanxi Province, Xi'an 710127, China}

\author{Kun \surname{Zhang}}
\email{kunzhang@nwu.edu.cn}
\affiliation{School of Physics, Northwest University, Xi’an 710127, China}
\affiliation{Shaanxi Key Laboratory for Theoretical Physics Frontiers, Xi'an 710127, China}
\affiliation{Peng Huanwu Center for Fundamental Theory, Xi'an 710127, China}
\affiliation{Fundamental Discipline Research Center for  Quantum Science and Technology of Shaanxi Province, Xi'an 710127, China}

\author{Vladimir \surname{Korepin}}
%\email{vladimir.korepin@stonybrook.edu}
\affiliation{C.N. Yang Institute for Theoretical Physics, Stony Brook University, New York 11794, USA}

\date{\today}

\begin{abstract}
Grover's algorithm provides a quadratic speedup over classical algorithms for searching an unstructured database and is known to be strictly optimal in oracle query complexity, with tight bounds on its success probability. Although the standard Grover search cannot be further accelerated in the full-search setting, a trade-off between accuracy and query complexity gives rise to the partial search problem. The Grover-Radhakrishnan-Korepin (GRK) algorithm is the standard and most extensively studied protocol for this task. In this work, we provide systematic numerical evidence that the GRK operator sequence gives the highest success probability in all examined cases, supporting it as the optimal ansatz among admissible compositions of global and local Grover operators. Guided by this numerically supported GRK ansatz, we derive an asymptotically tight upper bound on the maximal success probability within the GRK family and establish the corresponding lower bound on the minimal expected number of oracle queries. Furthermore, we investigate parallel quantum search within the partial-search framework. While a direct GRK-based parallelization does not outperform established parallel Grover schemes, we demonstrate that a hybrid strategy combining partial and full search protocols yields a strict, though subleading, improvement over the outer parallel Grover scheme. Our results clarify the fundamental limits of quantum partial search and its role in optimizing parallel quantum search algorithms.
\end{abstract}

\maketitle
\section{\label{sec:intro} Introduction}

Grover's algorithm locates a marked item in an unstructured database of size $N$ using only $\mathcal{O}(\sqrt{N})$ quantum oracle queries, thereby achieving a quadratic speedup over classical search algorithms that require $\mathcal{O}(N)$ queries \cite{grover1996fast,Grover1997needle}. Its oracle query complexity has been rigorously proven to be strictly optimal \cite{BBHT98,zalka1999grover}. Moreover, for any fixed number of oracle calls $k$, the exact tight upper bound on the achievable success probability is known \cite{Dohotaru2008ExactQL}. 

Operationally, Grover's algorithm consists of iterated applications of two reflection operators, which coherently rotate the system state toward the target state by amplifying its amplitude. This mechanism was later formalized within the broader framework of amplitude amplification \cite{grover1998quantum,brassard2000quantum}. Owing to the fundamental role of unstructured search and the generality of amplitude amplification, Grover's algorithm has found numerous applications, including global optimization \cite{Baritompa2005GroversQA}, Boolean satisfiability \cite{Alonso2022EngineeringTD}, and quantum machine learning \cite{biamonteQuantumMachineLearning2017a}. It has also been experimentally implemented across a variety of quantum computing platforms, providing empirical demonstrations of quantum speedup \cite{Chuang1998,zhang2021implementation,pokharel2023demonstration,pokharel2024better}.

Since Grover's algorithm is strictly optimal in oracle query complexity, any further acceleration of the full search task must necessarily involve a modification of the problem itself. A natural strategy is to trade accuracy for query complexity, leading to the formulation of the partial search problem. In this setting, the objective is not to identify the complete target string, but only a partial specification of it. For example, if the target is encoded in a 4-qubit state, one may seek to determine only the first two qubits. Equivalently, the database can be partitioned into disjoint blocks, and the task reduces to identifying the block that contains the target state. For this modified objective, the standard Grover search is no longer optimal. A more efficient protocol is given by the Grover-Radhakrishnan-Korepin (GRK) algorithm \cite{grover2004partial,korepin2005simple,Korepin2006opt}. 

The central idea of the GRK algorithm is the introduction of a local Grover operator, which restricts the diffusion operation to a subspace corresponding to individual blocks. By appropriately interleaving global (standard) and local Grover iterations, the algorithm amplifies the amplitude of the correct block more efficiently than a purely global search. The GRK framework has subsequently been extended to various generalizations, including binary partial search \cite{Korepin2007BinaryQS} and multi-target partial search \cite{Choi2007QuantumPS,zhong2009quantum,Zhang2017QuantumPS}.

The GRK algorithm proceeds in three stages: one first applies the global Grover operator $k_1$ times to the initial uniform superposition, followed by $k_2$ applications of the local Grover operator, and concludes with a single additional global Grover iteration \cite{grover2004partial,korepin2005simple,Korepin2006opt}. Because the global and local Grover operators do not commute, the ordering of these operations plays a crucial role in determining the overall performance of the algorithm. Comparative analyses of various operator sequences suggest that the global-local-global structure prescribed by GRK yields the best performance \cite{Korepin2006quest}. However, despite evidence supporting its optimality, a fully rigorous proof that this sequence is strictly optimal among all admissible compositions remains elusive. The difficulty stems from the underlying algebraic structure: both operators admit a representation within the $O(3)$ group \cite{korepin2006group}, and the interplay of rotations and reflections complicates a direct characterization of globally optimal sequences.

In this work, we systematically investigate the performance of quantum partial search under a fixed total number of oracle queries, denoted by $k$. By exhaustively enumerating all $2^k$ admissible compositions of global and local Grover operators, we provide strong numerical evidence that the GRK sequence maximizes the success probability for each fixed $k$ in all cases examined, including both small and large database sizes. This exhaustive analysis overcomes a limitation of previous studies, which examined only a restricted subset of operator orderings\cite{Korepin2006quest}.

Guided by the observed optimal structure, we derive an asymptotically tight analytical upper bound within the GRK family on the maximal success probability achievable with $k$ oracle calls, and establish a corresponding tight lower bound on the minimal expected number of oracle queries. These results characterize the fundamental efficiency limits of quantum partial search. Building on this optimal framework, we further investigate partial search in the context of parallel quantum search. We show that a straightforward GRK-based parallelization does not surpass established parallel Grover schemes \cite{Gingrich2000Generalized,Cremers2025Speeding}. Motivated by this observation, we introduce a hybrid parallel strategy that combines partial and full search protocols. We demonstrate that this hybrid scheme yields a strict, though subleading, improvement over the outer parallel Grover scheme.

%We demonstrate that this hybrid scheme achieves a strictly improved parallel efficiency compared with direct parallel Grover search, especially in the regime of a small number of quantum processors.

This paper is organized as follows. In Sec.~\ref{sec:quantum_search}, we review Grover's algorithm and the GRK algorithm and establish the necessary notation. In Sec.~\ref{sec:num_quest}, we present a numerical investigation of the optimality of GRK operator sequences, evaluating both the success probability and the expected number of oracle queries. Asymptotically tight analytical bounds for these quantities are derived in Sec.~\ref{sec:exact_bounds_partial}. In Sec.~\ref{sec:exact_bounds_parallel}, we study parallel search schemes based on the partial search framework and introduce a hybrid parallel strategy. The main results are summarized in Sec.~\ref{sec:conclusion}. Detailed numerical data supporting Sec.~\ref{sec:num_quest} are collected in Appendix~\ref{app:numerical_quest}, the comparison between the hybrid and outer parallel schemes is presented in Appendix \ref{sec:comparing_hybrid_outer}, and the comparison between the hybrid and inner parallel schemes is given in Appendix \ref{sec:comparing_hybrid_inner}.

\section{\label{sec:quantum_search} Quantum search algorithms}

In this section, we review Grover's algorithm for the full search problem and the GRK algorithm for the partial search problem in Secs. \ref{subsec:grover} and \ref{subsec:partial_search}, respectively.

\subsection{\label{subsec:grover} Grover's algorithm}

%\cite{grover1996fast,nielsenQuantumComputationQuantum2010}

Consider an unstructured database of size $N = 2^n$. The unstructured search problem is to locate the solution $t \in \{0,1\}^n$ by querying a Boolean oracle $f(x)$, which is defined as $f(t)=1$ and $f(x)=0$ for all $x \neq t$. Quantum search exploits quantum superposition. In the context of search, the algorithm has the initial state
\begin{equation}
\label{eq:initial_state_sn}
\ket{s_n} = \frac{1}{\sqrt N}\sum_{x=0}^{N-1}\ket{x},
\end{equation}
which represents an equal superposition of all states in the search space. The Boolean function $f(x)$ is then implemented as a unitary operator. Using the phase kickback technique \cite{nielsenQuantumComputationQuantum2010}, this yields the standard quantum oracle
\begin{equation}
O_t\ket{x}=(-1)^{f(x)}\ket{x},
\end{equation}
equivalent to $O_t=I-2\ket{t}\bra{t}$. Since the oracle operator alone cannot directly extract the target state $\ket{t}$ from the initial superposition $\ket{s_n}$, Grover's algorithm employs a second, target-independent operator known as the diffusion operator,
\begin{equation}
D_n=2\ket{s_n}\bra{s_n}-I.
\end{equation}
This operator is closely related to the $n$-qubit Toffoli gate and can be decomposed efficiently into single- and two-qubit gates \cite{BBCDMSSSW95}.

The Grover operator, which combines the quantum oracle $O_t$ and the diffusion operator $D_n$, is defined as
\begin{equation}
G_n=D_nO_t.
\end{equation}
The Grover operator amplifies the amplitude of the target state while suppressing the amplitudes of all non-target states, when applied to the initial state $\ket{s_n}$. Crucially, this process does not introduce any relative amplitude differences among the distinct non-target states. It is therefore convenient to introduce the normalized state $\ket{t^\perp} = \sum_{x\neq t}\ket{x}/{\sqrt{N-1}}$, assuming a unique target state. In the two-dimensional subspace spanned by $\ket{t}$ and $\ket{t^\perp}$, the initial state can be expressed as
\begin{equation}
\ket{s_n}=\sin\theta_1\ket{t}+\cos\theta_1\ket{t^\perp},
\end{equation}
where $\sin\theta_1=1/\sqrt{N}$. The action of $G_n$ corresponds to a rotation by an angle $2\theta_1$, with the matrix representation
\begin{equation}
G_n =
\begin{pmatrix}
\cos 2\theta_1 & -\sin 2\theta_1 \\
\sin 2\theta_1 & \cos 2\theta_1
\end{pmatrix}.
\end{equation}
After $k$ iterations, the probability of finding the target state $\ket{t}$ is
\begin{equation}
\label{eq:success_pr_of_Grover}
\pr_{x=t}(k)=\abs{\left\langle t\abs{G_n^k}s_n\right\rangle}^2=\sin^2((2k+1)\theta_1).
\end{equation}
The optimal number of oracle queries that maximizes the success probability is approximately
\begin{equation}
k \approx \left\lfloor \frac{\pi}{4\theta_1} \right\rfloor \approx \left\lfloor \frac{\pi}{4} \sqrt{N} \right\rfloor,
\end{equation}
showing that Grover's algorithm finds the target state with $\mathcal O(\sqrt{N})$ oracle queries, demonstrating a quadratic speedup over the classical $\mathcal O(N)$ requirement.

Given that the growth rate of the success probability diminishes as it approaches 100\%, terminating Grover's algorithm before $k = \left\lfloor \pi\sqrt{N}/4 \right\rfloor$ and restarting upon failure can yield a speedup, a strategy known as serial execution or punctuated quantum search algorithm \cite{BBHT98,Gingrich2000Generalized}. Suppose we measure the state after $k$ Grover iterations, where the success probability is $\pr_{x=t}(k)$. The probability of finding the target state after repeating this process $q$ times is $\pr(k)(1 - \pr(k))^{q-1}$. Hence, the expected number of iterations required to locate the target state is
\begin{equation}
\E_{x=t}(k)=k\sum_{q=1}^{\infty}q\pr(k)\left(1-\pr(k)\right)^{q-1}=\frac{k}{\Pr(k)}.
\end{equation}
The value $k_\text{min}$ that minimizes this expectation satisfies
\begin{equation}
\frac{\mathrm d}{\mathrm dk} \frac{k}{\sin^2((2k+1)\theta_1)} = 0,
\end{equation}
which is equivalent to the transcendental equation $\tan((2k_{\min}+1)\theta_1) = 4\theta_1 k_{\min}$. Numerically, one obtains $k_{\min} \approx 0.5828\sqrt{N}$, with a corresponding success probability $\pr_{x=t}(k_{\min}) \approx 0.8446$. The resulting minimal expected iteration number is
\begin{equation}
\label{eq:E_min_Grover}
\E^{\min}_{x=t}=\frac{k_{\min}}{\Pr(k_{\min})} \approx 0.69\sqrt{N},
\end{equation}
which is lower than the standard Grover expectation $\pi\sqrt{N}/4 \approx 0.7854\sqrt{N}$.

\subsection{\label{subsec:partial_search} Partial search algorithm}

In the full search problem, the goal is to locate the complete $n$-bit target string $t$. In the partial search setting, however, only partial information about $t$ is required. Specifically, the task is to find $t_1$ as $t$ is decomposed as $t = t_1 t_2$. Denote the bit lengths of $t_1$ and $t_2$ as $n-m$ and $m$, respectively. This is equivalent to partitioning the database into $K=2^{n-m}$ blocks with $b=2^m$ items in each block and identifying the block labeled by $t_1$. The GRK algorithm, which solves this partial search problem, employs an additional diffusion operator given by \cite{grover2004partial,korepin2005simple,Korepin2006opt}
\begin{equation}
D_{m} = I_{n-m}\otimes \left(2|s_m\rangle\langle s_m|-I_{m}\right),
\end{equation}
where $\ket{s_m}$ is the uniform superposition state over all states within a single block. This operator can be viewed as a rescaled version of the full diffusion operator $D_n$. In what follows, we refer to $D_n$ and $D_m$ as the global and local diffusion operators, respectively.        

The combination of the quantum oracle $O_t$ and the local diffusion operator $D_m$ defines the local Grover operator
\begin{equation}
G_{m} = D_{m}O_t.
\end{equation}
Unlike the two-dimensional subspace spanned by ${|t\rangle,|t^\perp\rangle}$ in the standard Grover search, the GRK algorithm admits a natural three-dimensional representation given by the orthonormal basis
\begin{subequations}
\begin{align}
    &|t\rangle = |t_1\rangle \otimes |t_2\rangle, \\
&|b\bar t\rangle = \frac{1}{\sqrt{b-1}} \sum_{j \neq t_2} |t_1\rangle \otimes |j\rangle, \\
&|\bar b\rangle = \frac{1}{\sqrt{N-b}}\left(\sqrt{N}|s_n\rangle - |t\rangle - \sqrt{b-1}|b\bar t\rangle\right).
\end{align}
\end{subequations}
Here, $|b\bar t\rangle$ is the uniform superposition of all non-target states within the target block, while $|\bar b\rangle$ is the uniform superposition of all states belonging to non-target blocks.

In the basis $\{|t\rangle,|b\bar t\rangle,|\bar b\rangle\}$, the initial state $|s_n\rangle$ defined in Eq. (\ref{eq:initial_state_sn}) can be expressed as
\begin{equation}
\label{eq:sn_three}
|s_n\rangle = \sin\gamma\sin\theta_2|t\rangle+\sin\gamma\cos\theta_2|b\bar t\rangle + \cos\gamma |\bar b\rangle,
\end{equation}
with the angles given by
\begin{equation}
\label{eq:define_theta2_gamma}
\sin\theta_2=\frac{1}{\sqrt b},\qquad \sin\gamma=\frac{1}{\sqrt K}.
\end{equation}
Both the global Grover operator $G_n$ and the local Grover operator $G_m$ can be represented in this basis as
\begin{widetext}
\begin{equation}
G_n = \begin{pmatrix}
1- 2\sin^2\gamma\sin^2\theta_2 & 2\sin^2\gamma\sin\theta_2\cos\theta_2 & 2\sin\gamma\cos\gamma\sin\theta_2 \\
-2\sin^2\gamma\sin\theta_2\cos\theta_2 & 2\sin^2\gamma\cos^2\theta_2 - 1 & 2\sin\gamma\cos\gamma\cos\theta_2 \\
-2\sin\gamma\cos\gamma\sin\theta_2 & 2\sin\gamma\cos\gamma\cos\theta_2 & 2\cos^2\gamma - 1 \\
\end{pmatrix},\quad
G_m = \begin{pmatrix}
\cos2\theta_2 & \sin2\theta_2 & 0 \\
-\sin2\theta_2 & \cos2\theta_2 & 0 \\
0 & 0 & 1 \\
\end{pmatrix}.\label{eq:matrix_GnGm}
\end{equation}
\end{widetext}
Therefore the algorithm given by $G_n$ and $G_m$ admits a representation within $O(3)$ group \cite{korepin2006group}. Note that $\det G_n = -1$. 

The GRK algorithm consists of three sequential steps: applying the global Grover operator $k_1$ times, followed by the local Grover operator $k_2$ times, and finally applying one global Grover operator. The probability of successfully identifying the target block $t_1$ is
\begin{equation}
\label{eq:pr_partial_k1_k2}
\pr_{x=t_1}(k_1,k_2) = 1-|\langle\bar{b} |G_n G_m^{k_2}G_n^{k_1}|s_n\rangle|^2.
\end{equation}
The goal is to find $k_1$ and $k_2$ such that $\pr_{x=t_1}(k_1,k_2) \approx 1$, while minimizing the total number of oracle queries $k_1+k_2+1$. In the limit of a large database and large block size ($n\gg 1$ and $m\gg 1$), the optimal iteration counts can be parameterized as
\begin{equation}
k_1 = \frac{\pi}{4} \sqrt{N} - \eta \sqrt{b}, \quad k_2 = \alpha \sqrt{b},\label{eq:j1j2_of_GRK}
\end{equation}
where $\eta$ and $\alpha$ are the free parameters. The condition $\pr_{x=t_1}(k_1,k_2) = 1$ leads to
\begin{equation}
\tan \frac{2\eta}{\sqrt{K}} = \frac{2\sqrt{K} \sin 2\alpha}{K - 4 \sin^2 \alpha}.
\end{equation}
The total oracle queries $k_1+k_2+1\approx \pi\sqrt{N}/4 + (\alpha - \eta) \sqrt{b}$ is minimized at \cite{Korepin2006opt}
\begin{equation}
\tan \frac{2\eta_K}{\sqrt{K}} = \frac{\sqrt{3K-4}}{K-2}, \quad \cos 2\alpha_K = \frac{K-2}{2(K-1)}.\label{eq:parameters_of_j1j2_of_GRK}
\end{equation}
In the large block number limit ($n-m\gg1$), the optimal parameters converge to $\alpha_K \to \pi/6$ and $\eta_K \to \sqrt 3 / 2$, which gives
\begin{equation}
\label{eq:GRK_k1_k2}
k_1 = \frac{\pi}{4} \sqrt{N} - \frac{\sqrt{3b}}{2}, \quad k_2 = \frac{\pi}{6} \sqrt{b}.
\end{equation}
We observe that the GRK algorithm is faster than Grover algorithm by a factor $c_\text{GRK}\sqrt b$ with $c_\text{GRK} = \sqrt 3/2-\pi/6\approx 0.3424$. 

It has been proven that the speedup of the quantum partial search algorithm over the standard Grover search is at most a factor of $\sqrt b$ \cite{grover2004partial}. The factor $\sqrt b$ is tight because any further improvement would imply the existence of a full search protocol (via a sequential partial search algorithm) faster than Grover's algorithm, which is impossible, given Grover's algorithm has been proven to be query-optimal. A key open question, however, is whether other operator sequences, such as interplay between global and local Grover operators, could yield a larger speedup factor $c\sqrt b$ with $c>c_\text{GRK}\approx 0.3424$. This raises the fundamental question of whether the GRK sequence is strictly optimal among all possible sequences. In the following, we compare admissible operator sequences and provide numerical evidence supporting the GRK sequence as the optimizing
ansatz in the cases examined.

%In the following, we comprehensively compare all admissible operator sequences to provide compelling evidence for the strict optimality of the GRK algorithm.

\section{\label{sec:num_quest} Numerical investigations on the optimal quantum partial search algorithm}

In this section, we first describe the setup for the numerical tests in Sec. \ref{sec:setup_numerical_test}. We then present, in Sec. \ref{sec:numerical_results}, the numerically obtained bounds on the success probability and the expected iteration number for the partial search algorithm.

\subsection{Numerical setups}

\label{sec:setup_numerical_test}

Previous studies have examined whether certain specific sequences could outperform the GRK algorithm \cite{Korepin2006quest}. For a given number of qubits $n$ and block size $m$, the numerical task of identifying the optimal sequence that achieves $\pr_{x=t_1}(k_1,k_2) \approx 1$ is complicated by the need to define a precise success probability threshold. Note that $\pr_{x=t_1}(k_1,k_2)$ is rarely exactly equal to 1 \cite{ye2025deterministic}. Here, we circumvent this numerical difficulty by considering a fixed total number of oracle queries $k_\text{tot}$ and comparing the success probabilities across all $2^{k_\text{tot}}$ possible sequences.

We define an operator sequence as
\begin{equation}
\label{eq:S_sequence}
S_{n,m}(k_1,k_2,\ldots,k_q)=
\begin{cases}
G_n^{k_1}G_m^{k_2}\cdots G_m^{k_{q-1}}G_n^{k_q},
& q \ \text{odd},\\[1mm]
G_m^{k_1}G_n^{k_2}\cdots G_m^{k_{q-1}}G_n^{k_q},
& q \ \text{even}.
\end{cases}
\end{equation}
Here \(k_1,k_2,\ldots,k_{q-1}\in\mathbb N^+\), while the last exponent is allowed to be \(k_q\in\mathbb N_0\). Thus the last index \(k_q\) always corresponds to a global Grover operator. To maintain a consistent notation, we require that the final index $k_q$ always corresponds to the number of global Grover operators applied. For example, \(S_{n,m}(1,0)=G_m\), $S_{n,m}(k,1) = G_m^{k}G_n$ and $S_{n,m}(1, k_2, k_1) = G_nG^{k_2}_m G_n^{k_1}$. Physically, a local Grover step \(G_m\) leaves the amplitude of the non-target block state \(|\bar b\rangle\) unchanged, and therefore a final \(G_m\) does not change the success probability for identifying the target block. Such sequences can thus be regarded as containing a redundant final oracle query. Nevertheless, for numerical convenience and to avoid introducing any bias in the enumeration, we scan all \(2^{k_{\mathrm{tot}}}\) binary sequences of length \(k_{\mathrm{tot}}\), where each oracle query is followed by either a global Grover step \(G_n\) or a local Grover step \(G_m\).

The total number of oracle queries is given by
\begin{equation}
k_{\text{tot}} = \sum_{p=1}^{q} k_p.
\end{equation}
The probability that the sequence $S_{n,m}(k_1, k_2, \cdots, k_q)$ successfully identifies the target block is
\begin{equation}
\pr_{x=t_1}(k_1, k_2, \cdots, k_q) = 1-\left| \langle \bar b | S_{n,m}(k_1, k_2, \cdots, k_q) | s_n \rangle \right|^2.
\end{equation}
Consequently, the expected number of iterations is
\begin{equation}
\E_{x=t_1}(k_1, k_2, \cdots, k_q) = \frac{k_{\text{tot}}}{\pr_{x=t_1}(k_1, k_2, \cdots, k_q)}.
\end{equation}
Then the problem is well-defined by finding
\begin{multline}
\label{eq:pr_max}
    \pr^\text{max}_{x=t_1}(k_{\text{tot}}) \\
    = \max_{k_1,\cdots,k_q}\left\{\pr_{x=t_1}(k_1, k_2, \cdots, k_q)|\sum_{p=1}^{q} k_p = k_{\text{tot}}\right\}.
\end{multline}
Then the minimal expected query can be obtained from
\begin{equation}
    \E^\text{min}_{x=t_1} = \min_{k_\text{tot}}\left\{\frac{k_\text{tot}}{\pr^\text{max}_{x=t_1}(k_{\text{tot}})}\right\}.
\end{equation}
Note that the last operator of optimal sequence giving $\pr^\text{max}_{x=t_1}(k_{\text{tot}})$ should be the global Grover operator because $\langle\bar b|G^k_m|\psi\rangle = \langle\bar b|\psi\rangle$.

\subsection{Numerical bounds and optimal sequences}

\label{sec:numerical_results}

\begin{figure}[htpb]
\centering
\includegraphics[width=0.5\textwidth]{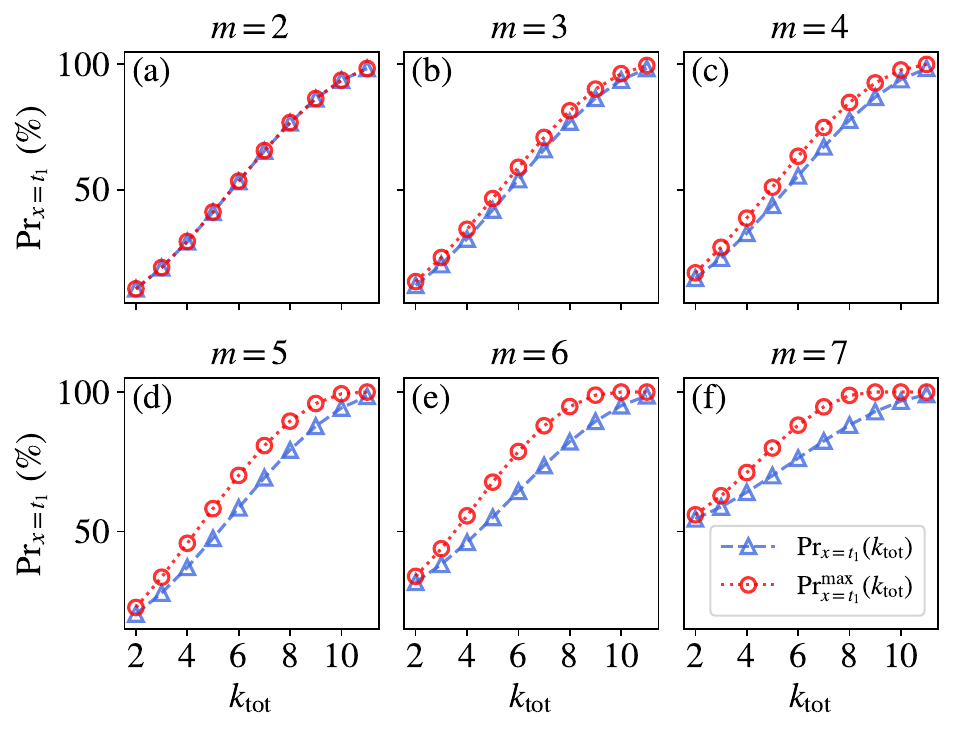}
\caption{Success probabilities of the quantum partial search algorithm as a function of the total number of oracle queries $k_\text{tot}$, for a database of size $N=2^8$. Values marked with red circles represent $\pr^\text{max}_{x=t_1}(k_{\text{tot}})$, as defined in Eq. (\ref{eq:pr_max}), obtained by optimizing over all possible operator sequences. Blue triangles correspond to the case where the partial search is performed using only the global Grover operator. The analytical expression for this probability is given in Eq. (\ref{eq:pr_grover}). The optimal sequences associated with $\pr^\text{max}_{x=t_1}(k_{\text{tot}})$ are listed in Table \ref{tab:n=8,pr}.}
\label{fig:pr_max}
\end{figure}

Consider a database of size $N=2^8$. Grover's algorithm requires $k=12$ oracle queries to reach a success probability of approximately $99.99\%$. We then focus on operator sequences $S_{n,m}$ comprising at most $11$ global or local Grover operators, i.e., with total query counts $k_\text{tot} \in [2, 11]$. The block size, determined by $m$, is varied within the range $m \in [2, 7]$. We conduct an exhaustive numerical search to identify, for each values $k_\text{tot}$ and $m$, the sequence yielding the maximal success probability $\pr^\text{max}_{x=t_1}(k_{\text{tot}})$ as defined in Eq. (\ref{eq:pr_max}). The resulting values of $\pr^\text{max}_{x=t_1}(k_{\text{tot}})$ are displayed in Fig. \ref{fig:pr_max}, and the corresponding optimal sequences are listed in Table \ref{tab:n=8,pr} in Appendix \ref{app:numerical_quest}. Except for a few over-rotation cases at large $k_\text{tot}$ discussed in Appendix \ref{app:numerical_quest}, the optimal sequences take the GRK form
\begin{equation}
S_{n,m}(1,k_2,k_1) = G_nG_m^{k_2}G_n^{k_1},
\end{equation}
providing compelling numerical evidence supporting the GRK sequence as the optimal ansatz in the cases examined. Additionally, when the database is partitioned into two blocks (corresponding to $m=n-1$ and $K=2$), we find that the optimal sequence omits the first global Grover operator, i.e., $k_1=0$. This result is consistent with the analysis presented in \cite{Korepin2006opt}.

In addition to the exhaustive enumeration for \(n=8\), we have performed targeted scans for larger database sizes up to \(n=40\) with a fixed total query count \(k_\text{tot}=20\). In every instance examined, the GRK‑form sequence yields the minimal expected number of iterations. Representative results are summarized in Appendix \ref{app:numerical_quest}.

\begin{figure}[htpb]
\centering
\includegraphics[width=0.5\textwidth]{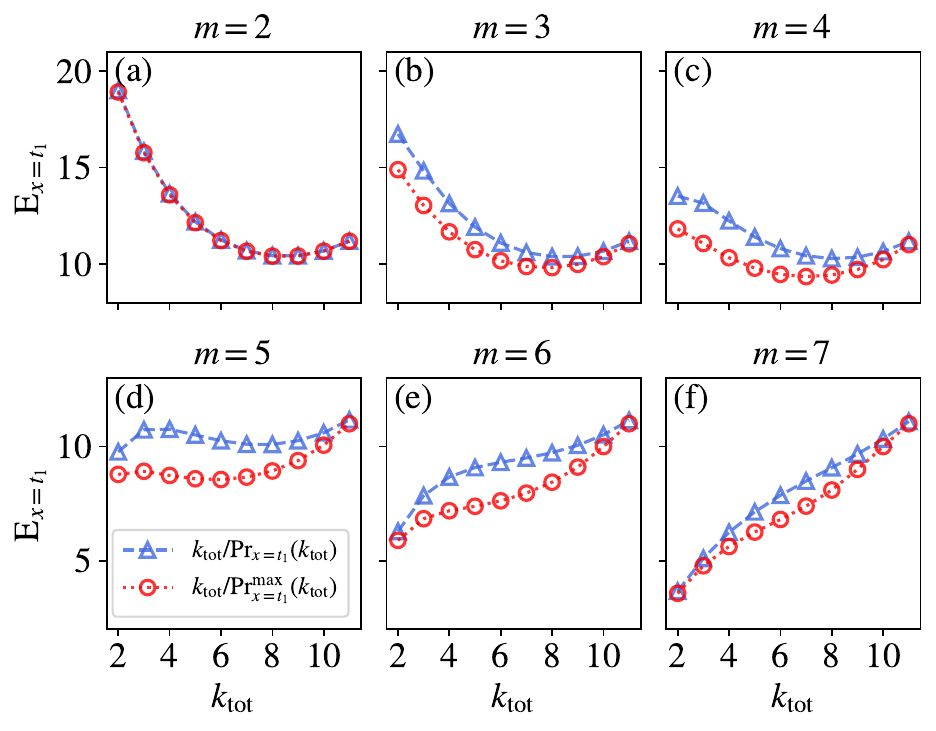}
\caption{Expected iteration numbers for the quantum partial search algorithm as a function of the total oracle queries $k_\text{tot}$, for a database of size $N=2^8$. The red circles show the values obtained from the optimized success probability $\pr^\text{max}_{x=t_1}(k_{\text{tot}})$ defined in Eq. (\ref{eq:pr_max}). The blue triangles correspond to using only the global Grover operator, with the success probability $\pr_{x=t_1}(k_{\text{tot}})$ given in Eq. (\ref{eq:pr_grover}). The corresponding numerical values are detailed in Table \ref{tab:n=8,k}.}
\label{fig:E}
\end{figure}

The corresponding expected numbers of iterations are shown in Fig. \ref{fig:E}, where the explicit values are listed in Table \ref{tab:n=8,k} in Appendix \ref{app:numerical_quest}. For $m \leq 5$, we find that the expected iteration number is not a monotonic function of $k_\text{tot}$. Therefore, a certain speedup can be achieved by measuring before $k_\text{tot}$ reaches its maximum value, analogous to the strategy employed in Grover's algorithm. In contrast, for $m > 5$, the minimal expected iteration number always occurs at $k_\text{tot}=2$, indicating low efficiency of the algorithm in this regime. This conclusion will be analytically verified in Sec. \ref{subsec:k_partial_search_with local}.

In Figs. \ref{fig:pr_max} and \ref{fig:E}, we also present the success probability and the expected iteration number for solving the partial search problem using only the global Grover operator. It is easy to see that
\begin{multline}
G_n^k|s_n\rangle = \sin((2k+1)\theta_1)|t\rangle\\ + \cos((2k+1)\theta_1)\frac{\sin\gamma\cos\theta_2}{\cos\theta_1}|b\bar t\rangle\\ + \cos((2k+1)\theta_1)\frac{\cos\gamma}{\cos\theta_1}|\bar b\rangle,
\end{multline}
due to 
\begin{equation}
    |t^\perp\rangle = \frac{\sin\gamma\cos\theta_2}{\cos\theta_1}|b\bar t\rangle + \frac{\cos\gamma}{\cos\theta_1}|\bar b\rangle,
\end{equation}
with the angles defined in Eq. (\ref{eq:define_theta2_gamma}). Then we have
\begin{equation}
\label{eq:pr_grover}
    \pr_{x=t_1}(k) = \sin^2((2k+1)\theta_1)+\frac{b-1}{N-1}\cos^2((2k+1)\theta_1).
\end{equation}
The first term gives the success probability of finding the full target state, while the second term corresponds to the success probability for the non-target states within the target block. Since the second term scales only as $\mathcal O(b/N)$, this indicates that Grover's algorithm (with only the global Grover operator) is not efficient for solving the partial search problem.

\section{\label{sec:exact_bounds_partial} \textbf{Asymptotically tight} bounds on quantum partial search algorithm}

In this section, we derive the analytical bounds for the quantum partial search algorithm. Tight bounds on the success probability and the expected iteration number are presented in Secs. \ref{subsec:pr_partial_search_with local} and \ref{subsec:k_partial_search_with local}, respectively.

\subsection{\label{subsec:pr_partial_search_with local} Upper bound on the success probability}

First, we address the upper bound on the success probability given by $k$ oracle queries. We have the following theorem. 
\begin{theorem}\label{theorem:pr_GRK}
Assuming that, for \(k_\mathrm{tot}=1+\alpha\sqrt{N}\), the optimal sequence is attained within the GRK family \(G_nG_m^{k_2}G_n^{k_1}\), the maximal target-block success probability in the joint limit \(b,K\gg1\) is
\begin{equation}
\label{eq:pr_GRK}
\pr^{\max}_{x=t_1}(k_\mathrm{tot})
=
\sin^2(2\alpha)
+
\epsilon\gamma\sin(4\alpha)
+
\mathcal O(\gamma^2)
+
\mathcal O(b^{-1/2}),
\end{equation}
where \(\epsilon\approx0.6849\).
\end{theorem}

%$k_\text{tot}=1+k_1+k_2=1+\alpha\sqrt{N}$

%with the constant $\epsilon=-2\gamma f\left(\pi/6\right),f\left(\pi/6\right) \approx -0.342427$.

\begin{proof}
Under the assumption stated in the theorem, supported by the numerical evidence in Sec.~\ref{sec:num_quest}, the optimization for a fixed total number of oracle queries $k_{\mathrm{tot}}$ can be restricted to sequences of the GRK form $G_nG_m^{k_2}G_n^{k_1}$. Consequently, the optimization reduces to finding the parameters $k_1$ and $k_2$ that maximize the success probability for a fixed $k_\text{tot}$. 

In the large block limit $b\gg 1$, the initial state $|s_n\rangle$ given by Eq. (\ref{eq:sn_three}) can be written as
\begin{equation}
|s_n\rangle = (0,\sin\gamma,\cos\gamma)^T + \mathcal O\left(\frac{1}{\sqrt b}\right).
\end{equation}
Acting $G_nG_m^{k_2}G_n^{k_1}$ on the initial state gives
\begin{multline}
\langle \bar b|G_nG_m^{k_2}G_n^{k_1}|s_n\rangle  = -\sin(2\gamma)\sin(2k_1\theta_1)\sin(2k_2\theta_2) \\ +\sin(2\gamma)\sin\gamma\cos(2k_1\theta_1)\cos(2k_2\theta_2)\\ + \cos\gamma\cos(2\gamma)\cos(2k_1\theta_1) + \mathcal O\left(\frac{1}{\sqrt b}\right).
\end{multline}
Furthermore, consider the large block number limit $K\gg 1$ (equivalent to $\gamma\to 0$), which gives
\begin{multline}
\langle \bar b|G_nG_m^{k_2}G_n^{k_1}|s_n\rangle = -2\gamma \sin(2k_1\theta_1)\sin(2k_2\theta_2)\\ + \cos(2k_1\theta_1) + \mathcal O\left(\gamma^2\right) + \mathcal O\left(\frac{1}{\sqrt b}\right).\label{eq:3rd_GRK_inistate}
\end{multline}

Since the total oracle queries $k_\mathrm{tot} = 1+\alpha\sqrt{N}$ is fixed, we have $k_\mathrm{tot}\theta_1 \approx \alpha$. Therefore, we have 
\begin{equation}
    k_1\theta_1 = (k_\text{tot}-k_2-1)\theta_1 \approx \alpha-k_2\theta_1.
\end{equation}
Recall that we have the angles $\theta_1\approx 1/\sqrt N$, $\theta_2\approx 1/\sqrt b$, and $\gamma\approx 1/\sqrt K$. From $N=bK$, we know $\theta_1\approx \theta_2\gamma$. Suppose that $k_2\theta_2 = \beta$, then we have
\begin{equation}
    k_1\theta_1 \approx \alpha - \beta\gamma. 
\end{equation}
Thus Eq. \eqref{eq:3rd_GRK_inistate} can be rewritten as
\begin{multline}
\langle \bar b|G_nG_m^{k_2}G_n^{k_1}|s_n\rangle = \cos(2\alpha) + 2\sin(2\alpha)\gamma(\beta - \sin(2\beta))\\ + \mathcal O\left(\gamma^2\right) + \mathcal O\left(\frac{1}{\sqrt b}\right).\label{eq:sqrtpr:u}
\end{multline}
Define the function
\begin{equation}
f(x) = x - \sin(2x),
\end{equation}
which takes the minimum at $x = \pi/6$. The minimal value is $f\left(\pi/6\right) \approx -0.342427$. Let $\epsilon=-2 f\left(\pi/6\right)\approx 0.6849$ for convenience. The success probability of finding the target block has the maximum
\begin{multline}
    \pr^{\max}_{x=t_1}(k_\mathrm{tot}) =  1- \left|\langle \bar b|G_nG_m^{k_2}G_n^{k_1}|s_n\rangle\right|^2 \nonumber \\
    =  \sin^2(2\alpha) + \sin(4\alpha)\epsilon\gamma + \mathcal O\left(\gamma^2\right)+\mathcal O(b^{-1/2}).
\end{multline}
\end{proof}

\begin{figure}[t]
\centering
\includegraphics[width=\columnwidth]{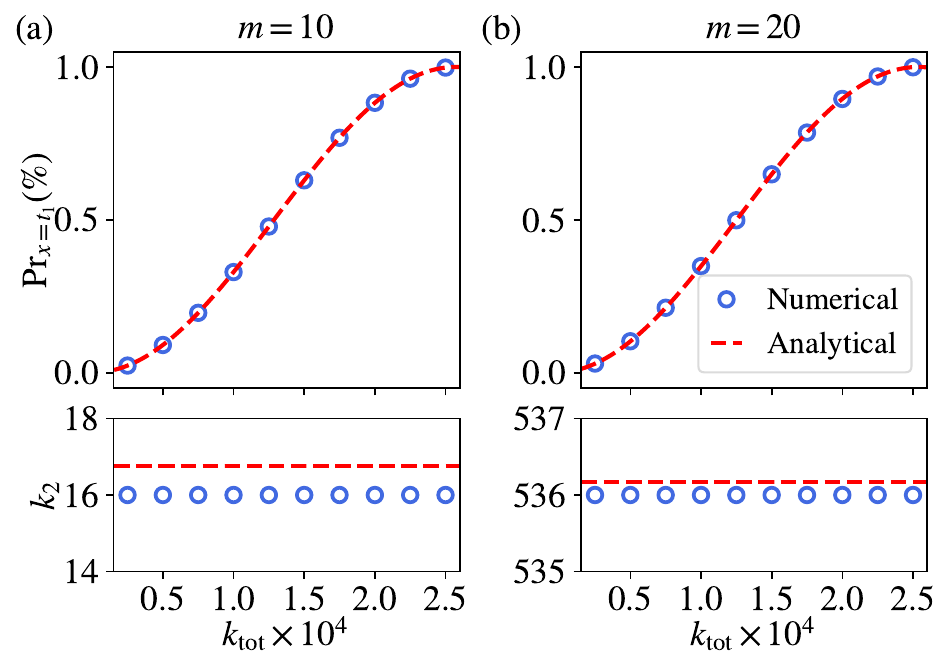}
\caption{Maximal success probability of the quantum partial search algorithm with (a) $m=10$ and (b) $m=20$, for a database of size $N=2^{30}$. The analytical curves are plotted from Eq. \eqref{eq:pr_GRK}. Numerical results are obtained by maximizing the success probability over $G_nG_m^{k_2}G_n^{k_1}$, under the constraint $k_\text{tot} = 1+k_1+k_2$. The value of $k_2$ used analytically is given by $k_2 = \pi\sqrt b/6$.}
\label{fig:bound_pr}
\end{figure}

The first term of the maximal success probability $\pr^{\max}_{x=t_1}(k_\mathrm{tot})$ in Eq. \eqref{eq:pr_GRK} is about the full target state $|t\rangle$. The second term is about the non-target state in the target block. Note that the second term scales as $1/\sqrt K$, which is optimal. Quantitatively, we can compare it with the full search. 
\begin{corollary}
For a fixed total number of oracle queries \(k_{\mathrm{tot}}\), the GRK partial-search sequence gives a higher target-block success probability than the standard full Grover search
\begin{equation}
    \pr^{\max}_{x=t_1}(k_\mathrm{tot})-\pr_{x=t}(k_\mathrm{tot})
    =
    \epsilon\gamma\sin(4\alpha)
    +\mathcal O(\gamma^2),
\end{equation}
where \(\epsilon\approx0.6849\).
\end{corollary}
\begin{proof}
The proof can be easily derived from comparing the maximal success probability $\pr^{\max}_{x=t_1}(k_\mathrm{tot})$ in Eq. \eqref{eq:pr_GRK} and the success probability of Grover search $\pr_{x=t}(k_\mathrm{tot})$ in Eq. (\ref{eq:success_pr_of_Grover}). 
\end{proof}

An interesting observation is that the maximal success probability is achieved at $k_2 = \pi\sqrt b/6$, which is independent of both the total database size $N$ and the total number of oracle queries $k_\text{tot}$. We also provide numerical verification of the analytical bound derived in Eq. \eqref{eq:pr_GRK}, as shown in Fig. \ref{fig:bound_pr}. The numerical results confirm that $k_2$ remains constant regardless of $k_\text{tot}$.

We emphasize that the expansion in Eq.~\eqref{eq:pr_GRK} is asymptotic in the joint limit \(b\gg1\) and \(K\gg1\). The nominal scales of the two remainder terms are \(1/K\) and \(1/\sqrt b\), respectively. Therefore, for small databases or small block sizes, these corrections are not necessarily negligible. The purpose of Eq.~\eqref{eq:pr_GRK} is to characterize the large-\(b\), large-\(K\) scaling behavior. Nevertheless, for the database size \(N=2^{30}\) considered in Fig.~\ref{fig:bound_pr}, the asymptotic expression shows good agreement with the numerical results.

\subsection{\label{subsec:k_partial_search_with local} Lower bound on the expected iteration number}

As shown in the previous subsection, the growth rate of the maximal success probability of the partial search algorithm with respect to the total number of oracle queries $k_\text{tot}$ decreases as the success probability approaches unity. It is therefore advantageous to terminate the algorithm before the success probability approaches unity. Employing a serial execution strategy can, in this context, provide a speedup. First, we aim to find the total number of oracle queries $k_\text{tot}$ which gives the minimal expected iteration number. 
\begin{lemma}
\label{lemma:extrem_point}
In the asymptotic regime $b \gg 1$ and $K \gg 1$, the GRK sequence $G_n G_m^{k_2} G_n^{k_1}$ achieves its minimal expected iteration number at
\begin{equation}
\label{eq:k_tot_min_E}
    k_\mathrm{tot} \approx 0.5829\sqrt N - 0.4969\sqrt b,
\end{equation}
or at the boundary value $k_\mathrm{tot} = 1$.
\end{lemma}
\begin{proof}
The maximal success probability $\pr^{\max}_{x=t_1}(k_\mathrm{tot})$ in Eq. \eqref{eq:pr_GRK} gives the expected iteration number
\begin{multline}
    \label{eq:calculate_E(GRK)}
    \E_{x=t_1}(1,k_2,k_1)=\frac{1+k_1+k_2}{\pr_{x=t_1}(1,k_2,k_1)}\\
     =\frac{2\alpha\sqrt N}{1-\cos (4\alpha)+2\epsilon\gamma\sin(4\alpha)}+\mathcal O\left(\gamma^2\right),
\end{multline}
where $\alpha\sqrt N\approx k_\text{tot}$. To find its minimum, let 
\begin{equation}
\label{eq:g_alpha}
g(\alpha)=\frac{2\alpha}{1-\cos (4\alpha)+2\epsilon\gamma\sin(4\alpha)}.
\end{equation}
Let $g'(\alpha)=0$ which gives
\begin{multline}
1-\cos(4\alpha)-4\alpha\sin(4\alpha)\\
+2\epsilon\gamma\sin(4\alpha)-8\epsilon\gamma\alpha\cos(4\alpha)=0.\label{eq eq of alpha}
\end{multline}
To solve the above equation, we employ a perturbation approach. In the case $\epsilon = 0$, the solution within the interval $(0,1)$ is $\alpha_0 \approx 0.5829$. Since $0<\gamma \ll 1$ (corresponding to the large block number limit $K\gg 1$), we expand the solution as $\alpha_1 = \alpha_0 + \varepsilon\gamma$. Substituting this ansatz into Eq. \eqref{eq eq of alpha} yields $\varepsilon \approx -0.4969$, which determines the total iteration count in Eq. (\ref{eq:k_tot_min_E}). An additional solution of Eq. \eqref{eq eq of alpha} is $\alpha_2 = 0$. This corresponds to the case $k_\text{tot} = 1$.

The perturbative expansion is well controlled because the stationary equation has a simple root at \(\gamma=0\). Indeed, the derivative of the left-hand side of Eq.~\eqref{eq eq of alpha} with respect to \(\alpha\), evaluated at \((\alpha_0,\gamma=0)\), is nonzero. Therefore, the solution depends smoothly on \(\gamma\) in a neighbourhood of \(\gamma=0\). The next correction to \(\alpha_1\) is of order \(\gamma^2\). Since \(\gamma=1/\sqrt K\), this correction is of order \(1/K\), and its contribution to the expected iteration number is beyond the order retained in Eq.~\eqref{eq:k_tot_min_E}.

\end{proof}

Based on the above lemma, we obtain the minimal expected iteration number of the quantum partial search algorithm as follows.
\begin{theorem}\label{theorem:k_opt_GRK_partial_search}
Within the GRK-family optimization, the minimal expected number of iterations is
\begin{equation}
\label{eq:expected_k_GRK}
\E_{x=t_1}^{\min}
\approx
\begin{cases}
0.69\sqrt N - 0.4054\sqrt b, & m\leq \lfloor n/2\rfloor,\\[1mm]
K-\dfrac{8K^2}{N}, & m>\lfloor n/2\rfloor.
\end{cases}
\end{equation}
For \(m\leq \lfloor n/2\rfloor\), the minimum is attained by the GRK partial search operator, whereas for \(m>\lfloor n/2\rfloor\) it is achieved by \(k_{\mathrm{tot}}=1\).
\end{theorem}

\begin{proof}
Lemma \ref{lemma:extrem_point} gives two candidate values of $k_\text{tot}$, one of which provides the minimal expected iteration number. Note that the leading term of $k_\text{tot}$ in Eq. \eqref{eq:k_tot_min_E} corresponds precisely to the iteration number that minimizes the expected iteration number of the standard Grover search, as shown in Eq. \eqref{eq:E_min_Grover}. Consequently, we can estimate the minimal expected iteration number as $\E_{x=t_1} \approx 0.69\sqrt N$. The other extremum, $k_\text{tot}=1$, corresponds to a sequence containing only a single global Grover operator, which effectively reduces to random guessing. For this case, the expected iteration number satisfies $\E_{x=t_1} \approx K$. Recalling that $N=2^n$ and $K=2^{n-m}$, it follows that for $m \lessapprox n/2+0.5353$, the minimal expected iteration number $\E_{x=t_1} \approx 0.69\sqrt N$ is smaller than $K$.

The reasoning above indicates that for $m \leq \left\lfloor n/2 \right\rfloor$, the minimum is attained at the total query count given in Eq \eqref{eq:k_tot_min_E}. Using Eq. \eqref{eq:calculate_E(GRK)}, we can express it as $\E^\text{min}_{x=t_1} = g(\alpha)\sqrt N$, where the function $g(\alpha)$ is defined in Eq. \eqref{eq:g_alpha}. The expansion corresponding to the minimum, located at $\alpha_1 = \alpha_0+\varepsilon\gamma$, is as follows
\begin{multline}
    g(\alpha_1) = \frac{\alpha_0}{\sin^2(2\alpha_0)} \\
    + \left(\frac{\varepsilon}{\sin^2(2\alpha_0)} - \frac{2\alpha_0(2\varepsilon+\epsilon)\cos(2\alpha_0)}{\sin^3(2\alpha_0)}\right)\gamma + \mathcal O\left(\gamma^2\right).
\end{multline}
Substituting $\alpha_0 = 0.5829$, $\varepsilon = -0.4969$, and $\epsilon = 0.6849$ gives the minimal expected iteration number shown in Eq. \eqref{eq:expected_k_GRK}.

As for $k_\text{tot}=1$, from Eq. (\ref{eq:pr_grover}), we know the success probability is
\begin{equation}
    \pr_{x=t_1}(1) = 1 - \cos^2(3\theta_1)\frac{\cos^2\gamma}{\cos^2\theta_1}.
\end{equation}
To leading order, we find $\pr_{x=t_1}(1)\approx 8/N + 1/K$. Therefore we obtain $\E_{x=t_1}^{\min}\approx K - 8K^2/N$. 
\end{proof}

\begin{figure}[t]
\centering
\includegraphics[width=\columnwidth]{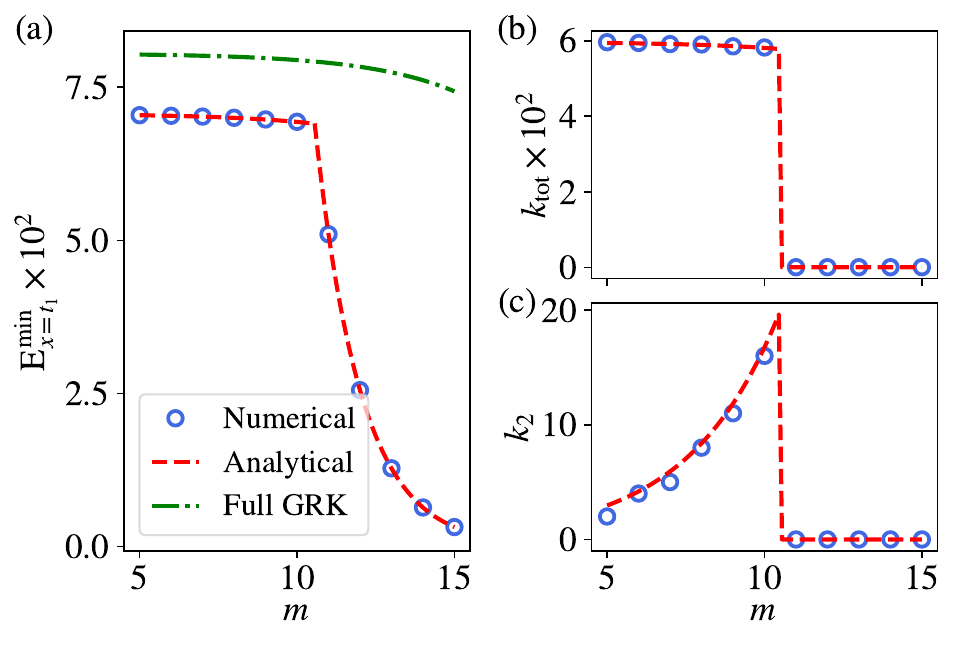}
\caption{(a) Minimal expected iteration number of the quantum partial search algorithm, and the corresponding iteration numbers (b) $k_\text{tot}$ and (c) $k_2$. The red dashed lines show the analytical results established in Theorem \ref{theorem:k_opt_GRK_partial_search}. The green dash-dotted line represents the GRK algorithm with success probability close to unity. The parameter is set as $n=20$.}
\label{Fig:E_min}
\end{figure}

As previously established in Sec. \ref{subsec:partial_search}, the GRK algorithm achieves a success probability approaching one with an iteration number $k_1$ and $k_2$ given by Eq. \eqref{eq:GRK_k1_k2}. The corresponding expected iteration number is $\E_{x=t_1} \approx \pi\sqrt N/4 - 0.3424\sqrt b$. Comparing it with the minimal expected iteration number derived in Eq. \eqref{eq:expected_k_GRK}, a speedup of approximately $0.0954\sqrt N+0.063\sqrt b$ is observed.  The results established in Theorem \ref{theorem:k_opt_GRK_partial_search} are numerically verified in Fig. \ref{Fig:E_min}, which clearly demonstrates a speedup compared to the GRK with success probability approaching one, and confirms a phase-transition-like behavior near $m=n/2$.

However, Theorem~\ref{theorem:k_opt_GRK_partial_search} also reveals that when the number of blocks \(K\) becomes smaller than approximately \(\sqrt{N}\), the boundary strategy \(k_{\mathrm{tot}}=1\), namely a single global Grover operation followed by measurement, can already achieve a smaller expected number of iterations than the interior GRK-based partial-search branch. In this regime, the problem is effectively close to guessing among \(K\) blocks, because the number of possible blocks has become sufficiently small. The apparent abrupt change near \(m=n/2\) therefore originates from the competition between two distinct optimization strategies: the GRK partial-search algorithm and the random guessing strategy.

\section{\label{sec:exact_bounds_parallel}Applications to quantum parallel search}

In this section, we begin by reviewing the inner and outer parallelization schemes of the quantum search algorithm in Sec. \ref{subsec:direct_parallel}. We then introduce a parallel scheme based on the partial search framework and analyze its efficiency in Sec. \ref{subsec:block_parallel}. Finally, we present an improved hybrid parallel strategy in Sec. \ref{subsec:mixed_parallel}.

\subsection{Inner and outer parallel search schemes}

\label{subsec:direct_parallel}

Suppose that we have $l$ QPUs (quantum processing unit) and wish to parallelize the search algorithm. Existing parallelization schemes for quantum search can be broadly categorized into two types, as described in \cite{kim2018time}.
\begin{definition}
\label{def:inner}
    Inner parallel scheme: The search space is partitioned into $l$ equally sized subsets, with each of the $l$ QPUs tasked with searching its assigned subset for the target state.
\end{definition}
\begin{definition}
\label{def:outer}
    Outer parallel scheme: Each of the $l$ QPUs runs the same Grover's algorithm, with the overall process halting as soon as any one processor finds the target state.
\end{definition}

The inner parallel search has a speedup because each quantum computer runs a search algorithm with the search space $N/l$. The target state is in one of $l$ sets. Therefore the expected iteration number is \footnote{When evaluating the speedup of a parallel quantum search scheme, the overall algorithm runtime can be approximated by the number of oracle queries executed on a single QPU. Thus, the term expected iteration number for a parallel search refers to the oracle count performed by one QPU. } 
\begin{equation}
\label{eq:E_inner}
\E_{\text{inner},\parallel}(k) = \frac{k}{\sin^2((2k+1)\theta_1')},
\end{equation}
with $\sin\theta_1' = \sqrt{l/N}$. The minimum, similar as Eq. \eqref{eq:E_min_Grover}, is given by
\begin{equation}
\label{eq:E_min_inner}
    \E^{\text{min}}_{\text{inner},\parallel}\approx 0.69\sqrt{\frac{N}{l}}.
\end{equation}
The inner parallel scheme is also called the multi-programming Grover's search algorithm \cite{park2023quantum}. Taking $l=2$ as an example, the quantum circuit is
$$
\begin{quantikz}[column sep=0.25cm, row sep=0.25cm]
\lstick[2]{1st QPU} \wireoverride{n}  & \wireoverride{n}  & \wireoverride{n}  & \wireoverride{n}  & \wireoverride{n}  & \lstick{$|0\rangle$} \wireoverride{n} & & \gate[2]{G_{n-1}^{k}} &  & \wireoverride{n}  \\
\wireoverride{n}  & \wireoverride{n}  & \wireoverride{n}  & \wireoverride{n}  & \wireoverride{n} & \lstick{$|0\rangle^{\otimes n-1}$} \wireoverride{n} & \gate{H^{\otimes n-1}} & & \meter{n-1} & \wireoverride{n} \\
\lstick[2]{2nd QPU} \wireoverride{n}  & \wireoverride{n}  & \wireoverride{n}  & \wireoverride{n}  & \wireoverride{n} & \lstick{$|1\rangle$} \wireoverride{n} & & \gate[2]{G_{n-1}^{k}} & & \wireoverride{n} \\
\wireoverride{n}  & \wireoverride{n}  & \wireoverride{n}  & \wireoverride{n}  & \wireoverride{n} & \lstick{$|0\rangle^{\otimes n-1}$}\wireoverride{n} & \gate{H^{\otimes n-1}} & & \meter{n-1} & \wireoverride{n} \\
\end{quantikz}
$$
Here $H$ is the single-qubit Hadamard gate \cite{nielsenQuantumComputationQuantum2010}. Note that the first qubit, being in the state $|0\rangle$ or $|1\rangle$, stands for the two partition sets. Although we can partition the search space, the oracle remains the original oracle operator. Only the diffusion operator has been rescaled to act on $n-1$ qubits. 

In the outer parallel scheme, each QPU conducts an independent and simultaneous search over the entire database \cite{Gingrich2000Generalized}. The search is considered successful as soon as any QPU locates a target state. Consequently, the overall success probability is $1-\left(1-\pr(k)\right)^l$, where the success probability of Grover’s algorithm $\pr(k)$ is given by Eq.\eqref{eq:success_pr_of_Grover}. Thus, the expected iteration number for the above algorithm is
\begin{equation}
\label{eq:E_outer}
    \E_{\text{outer},\parallel}(k) = \frac{k}{1-\left(1-\sin^2((2k+1)\theta_1)\right)^l}.
\end{equation}
When $l$ is large, the iteration number $k$ per QPU typically remains small. Under this condition, the expected iteration number for the outer parallel scheme can be approximated as
\begin{equation}
\E_{\text{outer},\parallel}(x) \approx \frac{x}{2(1-e^{-x^2})}\sqrt{\frac{N}{l}},
\end{equation}
where $x=(1+2k)\sqrt{l/N}$. Correspondingly the minimal expected iteration number is
\begin{equation}
\label{eq:E_outer_parallel}
    \E^\text{min}_{\text{outer},\parallel} \approx0.7835\sqrt{\frac{N}{l}}. 
\end{equation}
As the example on $l=2$, the corresponding quantum circuit is
$$
\begin{quantikz}[column sep=0.27cm]
\lstick{1st QPU} \wireoverride{n} & \wireoverride{n}  & \wireoverride{n} \wireoverride{n}  & \wireoverride{n} & \lstick{$|0\rangle^{\otimes n}$} \wireoverride{n} & \gate{H^{\otimes n}} & \gate{G_n^k} & \meter{n} \\
\lstick{2nd QPU} \wireoverride{n} & \wireoverride{n}  & \wireoverride{n} \wireoverride{n}  & \wireoverride{n} & \lstick{$|0\rangle^{\otimes n}$} \wireoverride{n} & \gate{H^{\otimes n}} & \gate{G_n^k} & \meter{n} \\
\end{quantikz}
$$
Although the inner parallel search is faster than the outer parallel search, the outer scheme is more flexible since $l$ could take any value. The inner scheme can only have $l = 2^r$ with $r\in \mathbb{N}^+$ and \(l \le N\).

\subsection{GRK‑based parallel search scheme}

\label{subsec:block_parallel}

The essence of the inner parallel scheme is to use Grover's algorithm to extract partial information about the target state. This naturally raises the question of whether the quantum partial search algorithm can be employed in parallel to achieve a further speedup.
\begin{definition}
\label{def:GRK}
    GRK‑based parallel scheme: Assume that $l = n/(n-m)$ is an integer. Each of the $l$ QPUs employs the GRK algorithm to determine $n-m$ bits of the target string.
\end{definition}
Recall that the GRK algorithm, with its local diffusion operator acting on $m$ qubits, identifies the $n-m$ bits of the target string (i.e., the qubits not acted upon by the local diffusion operator). These $n-m$ bits serve to label the target block. The specific choice of which $n-m$ bits define the block is, however, arbitrary. If each QPU successfully identifies a distinct set of $n-m$ bits, their results can be assembled to reconstruct the full solution. Consequently, the overall success probability for the parallel scheme is $\pr^l_{x=t_1}(k_1,k_2)$, where $\pr_{x=t_1}(k_1,k_2)$ is given by Eq. (\ref{eq:pr_partial_k1_k2}). The corresponding expected iteration number is then
\begin{equation}
    \E_{\text{GRK},\parallel}(k_\text{tot}) = \frac{1+k_1+k_2}{\pr^l_{x=t_1}(k_1,k_2)},
\end{equation}
which has the minimum
\begin{equation}
    \E^\text{min}_{\text{GRK},\parallel} = \min_{k_\text{tot}}\left\{\E_{\text{GRK},\parallel}(k_\text{tot})\right\}.
\end{equation}
Here the minimization is to find $k_\text{tot}$ and the corresponding portion of $k_1$ and $k_2$ which minimizes $\E_{\text{GRK},\parallel}$. 

The quantum circuit diagram with $l=2$ is (assuming that $n$ is an even number)
$$
\begin{quantikz}[column sep=0.2cm, row sep=0.25cm]
\lstick[2]{1st QPU} \wireoverride{n} & \wireoverride{n} & \wireoverride{n}  & \wireoverride{n}  & \wireoverride{n}  & \wireoverride{n}  & \lstick{$|0\rangle^{\otimes n/2}$} \wireoverride{n} & \gate{H^{\otimes{n/2}}} & \gate[2]{G_n^{k_1}} & \qw & \gate[2]{G_n} & \meter{n/2} & \wireoverride{n}  \\
\wireoverride{n} & \wireoverride{n} & \wireoverride{n}  & \wireoverride{n}  & \wireoverride{n}  & \wireoverride{n} & \lstick{$|0\rangle^{\otimes n/2}$} \wireoverride{n} & \gate{H^{\otimes{n/2}}} & & \gate{G_{n/2}^{k_2}} & & & \wireoverride{n} \\
\lstick[2]{2nd QPU} \wireoverride{n} & \wireoverride{n} & \wireoverride{n}  & \wireoverride{n}  & \wireoverride{n}  & \wireoverride{n} & \lstick{$|0\rangle^{\otimes n/2}$} \wireoverride{n} & \gate{H^{\otimes{n/2}}} & \gate[2]{G_n^{k_1}} & \gate{G_{n/2}^{k_2}} & \gate[2]{G_n} & & \wireoverride{n} \\
\wireoverride{n} & \wireoverride{n} & \wireoverride{n}  & \wireoverride{n}  & \wireoverride{n}  & \wireoverride{n} & \lstick{$|0\rangle^{\otimes n/2}$}\wireoverride{n}  & \gate{H^{\otimes{n/2}}} & & & & \meter{n/2} & \wireoverride{n} \\
\end{quantikz}
$$
The first QPU finds the first half bits of the target string, while the second QPU finds the rest bits. 

As the efficiency of GRK‑based parallel search algorithm, we compare it with the outer parallel search algorithm and have the following theorem.
\begin{theorem}
\label{thm:grk_parallel}
    At the large database limit $n\gg 1$, outer parallel search algorithm always outperforms GRK‑based parallel search algorithm for any $l$, namely $\E^\mathrm{min}_{\mathrm{outer},\parallel}<\E^\mathrm{min}_{\mathrm{GRK},\parallel}$.
\end{theorem}
\begin{proof}
    From Theorem \ref{theorem:pr_GRK}, we obtain the expected iteration number of GRK‑based parallel search algorithm
    \begin{equation}
        \E_{\text{GRK},\parallel}(\alpha) \approx \frac{\alpha\sqrt N}{\left(\sin^2(2\alpha) + \epsilon\gamma\sin(4\alpha)\right)^l},
    \end{equation}
    with $k_1+k_2=\alpha\sqrt N$. Here the angle $\gamma$ defined in Eq. \eqref{eq:define_theta2_gamma} equals to $\gamma\approx 1/\sqrt{2^{n/l}}$ due to $l = n/(n-m)$. We aim to find the minimal expected iteration number of GRK‑based parallel search algorithm given by
    \begin{equation}
        \E^\text{min}_{\text{GRK},\parallel} =\min_{\alpha} \left\{\frac{\alpha\sqrt N}{\left(\sin^2(2\alpha) + \epsilon\gamma\sin(4\alpha)\right)^l}\right\}.
    \end{equation}
    Below we analytically derive $\E^\text{min}_{\text{GRK},\parallel}$ for the two extreme cases, i.e., $l=2$ and $l=n$. Values for other $l$ can then be extrapolated accordingly.

    As for $l=2$ (equivalent to $m=n/2$), we have
    \begin{equation}
        \E_{\text{GRK},\parallel}(l=2) \approx \frac{\alpha\sqrt N}{\left(\sin^2(2\alpha) +  2^{-n/4} \epsilon\sin(4\alpha)\right)^2}.
    \end{equation}
    Note that the term $2^{-\frac n 4} \epsilon$ (with $\epsilon\approx 0.6849$) is small in the large database limit $n\gg 1$. Then we can adopt the perturbation method to find the minimum. Specifically, in the first leading order, the minimum is given by
    \begin{equation}
        \alpha^* = \alpha_0 - \frac{64\alpha_0^2+1}{128\alpha_0^2-6}\frac{\epsilon}{2^{n/4}} + \mathcal O\left(\frac{1}{2^{n/2}}\right),
    \end{equation}
    where $\alpha_0$ is given by $\tan(2\alpha_0) = 8\alpha_0$. Correspondingly, the minimal expected iteration number is
    \begin{equation}
        E^\text{min}_{\text{GRK},\parallel}(l=2)  = \frac{2^{n/2}}{\sin^4(2\alpha_0)}\left(\alpha_0-\frac{\epsilon}{2^{n/4+1}}\right) + \mathcal O\left(1\right).
    \end{equation}
    The approximated numerical values are $\alpha^* \approx 0.696-0.3913/\sqrt[4]{N}$ and
    \begin{equation}
        E^\text{min}_{\text{GRK},\parallel}(l=2)  \approx 0.7422\sqrt N - 0.3651\sqrt[4]{N}.
    \end{equation}
    Comparing it with the outer parallel efficiency $\E^\text{min}_{\text{outer},\parallel}$, we always find that  $E^\text{min}_{\text{GRK},\parallel}(l=2)>\E^\text{min}_{\text{outer},\parallel}(l=2)$ in the large database limit.

    As for $l=n$ (equivalent to $m=n-1$), we have
    \begin{equation}
        \E_{\text{GRK},\parallel}(l=n) \approx \frac{\alpha\sqrt N}{\left(\sin^2(2\alpha) +  2^{-1/2} \epsilon\sin(4\alpha)\right)^l}.
    \end{equation}
    As $l$ increases, the minimum is given by larger $\alpha$ implying a low efficiency of parallelism. In the large database limit $n\gg 1$, we have
    \begin{equation}
        E^\text{min}_{\text{GRK},\parallel}(l=n)  \approx \left(0.5433-\frac{1}{4\pi n}\right)\sqrt N.
    \end{equation}
    Apparently we have $E^\text{min}_{\text{GRK},\parallel}(l=n)>\E^\text{min}_{\text{outer},\parallel}(l=n)$. Thus we conclude that outer parallel search algorithm always outperforms GRK‑based parallel search algorithm.
\end{proof}
One reason the GRK‑based parallel search algorithm discussed here fails to improve existing results is requiring all QPUs to successfully find the partial target bits.

\subsection{Hybrid parallel search scheme}

\label{subsec:mixed_parallel}

Theorem \ref{theorem:pr_GRK} indicates that the maximal success probability of the quantum partial search can be decomposed into two components: the probability of finding the full target state and the probability of locating a non-target state within the target block. This reveals that the GRK algorithm inherently retains the capability to identify the full target state. Motivated by this property, we can combine the outer parallel scheme with the GRK-based parallel scheme, leading to what we term the hybrid parallel search scheme.
\begin{definition}
\label{def:hybrid}
    Hybrid parallel search scheme: Assume that $l = n/(n-m)$ is an integer. Each of the $l$ QPUs executes the GRK algorithm. The measurement outcome of each QPU is then verified using a classical oracle to determine whether it corresponds to the target state. Simultaneously, the block-label qubits obtained from all QPUs are combined, and the resulting candidate string is also checked via a classical oracle to confirm whether it matches the full target state.
\end{definition}

The essential difference between the hybrid scheme and the outer parallel scheme is not the GRK operator itself, but the additional verification channel created by combining partial information from different QPUs. In the outer parallel scheme, a run is successful only if at least one QPU directly outputs the full target string. In the hybrid scheme, there is an additional possibility: even if no individual QPU outputs the full target string, the block information obtained from different QPUs may be concatenated into the correct full string and then verified by the classical oracle. This extra block-concatenation channel is the structural origin of the hybrid improvement.

It is worth noting that the hybrid parallel search scheme requires that \(l\) divides \(n\); by contrast, the outer parallel scheme is completely free of this restriction. If implemented as explicit bitwise comparisons, the cost is \(\mathcal{O}(ln)\), which remains negligible compared with the \(\mathcal{O}(\sqrt{N})\) quantum oracle-query cost in the asymptotic comparison.

The quantum circuit for the hybrid parallel search algorithm is essentially the same as that of the GRK‑based parallel search, with the only modification being that all qubits are measured. As an example, for $l=2$, the corresponding quantum circuit diagram is shown below.
$$
\begin{quantikz}[column sep=0.2cm, row sep=0.4cm]
\lstick[2]{1st QPU} \wireoverride{n}  & \wireoverride{n} & \wireoverride{n} & \wireoverride{n}  & \wireoverride{n}  & \wireoverride{n}  & \lstick{$|0\rangle^{\otimes n/2}$} \wireoverride{n} & \gate{H^{\otimes n/2}} & \gate[2]{G_n^{k_1}} & \qw & \gate[2]{G_n} & \meter{n/2} & \wireoverride{n}  \\
\wireoverride{n} & \wireoverride{n} & \wireoverride{n}  & \wireoverride{n}  & \wireoverride{n}  & \wireoverride{n} & \lstick{$|0\rangle^{\otimes n/2}$} \wireoverride{n} & \gate{H^{\otimes n/2}} & & \gate{G_{n/2}^{k_2}} & & \meter{n/2} & \wireoverride{n} \\
\lstick[2]{2nd QPU} \wireoverride{n} & \wireoverride{n} & \wireoverride{n}  & \wireoverride{n}  & \wireoverride{n}  & \wireoverride{n} & \lstick{$|0\rangle^{\otimes n/2}$} \wireoverride{n} & \gate{H^{\otimes n/2}} & \gate[2]{G_n^{k_1}} & \gate{G_{n/2}^{k_2}} & \gate[2]{G_n} & \meter{n/2} & \wireoverride{n} \\
\wireoverride{n} & \wireoverride{n} & \wireoverride{n}  & \wireoverride{n}  & \wireoverride{n}  & \wireoverride{n} & \lstick{$|0\rangle^{\otimes n/2}$}\wireoverride{n}  & \gate{H^{\otimes n/2}} & & & & \meter{n/2} & \wireoverride{n} \\
\end{quantikz}
$$
The measured $n$-bit outcomes from the first or second QPUs are separately submitted to a classical oracle for verification. Additionally, the first half ($n/2$ bits) measured from the first QPU is combined with the second half ($n/2$ bits) measured from the second QPU, and this composite $n$-bit string is also checked via the classical oracle. 

Accordingly, the success probability of one run of the hybrid parallel search
can be obtained by separating the measurement outcomes of each QPU into three
exclusive cases. A QPU outputs the full target state with probability
\(\pr_{x=t}(k_1,k_2)\); it outputs a non-target state inside the target block
with probability
\(\pr_{x=t_1}(k_1,k_2)-\pr_{x=t}(k_1,k_2)\); and it outputs a state outside
the target block with probability \(1-\pr_{x=t_1}(k_1,k_2)\). Thus the hybrid scheme succeeds either if at least one QPU directly outputs the full target state, or if no QPU outputs the full target state but all QPUs still output states in their corresponding target blocks, so that the block information can be concatenated into the correct target string. These two events are disjoint. Therefore, the expected iteration number of the hybrid parallel search is
\begin{multline}
    \E_{\text{hybrid},\parallel}(k_\text{tot}) \\
    =
    \frac{1+k_1+k_2}
    {1-\left(1-\pr_{x=t}\right)^l
    +\left(\pr_{x=t_1}-\pr_{x=t}\right)^l}.
\end{multline}
Note that the success probability of identifying the
target block in the GRK algorithm \(\pr_{x=t_1}(k_1,k_2)\) is given by Eq.~\eqref{eq:pr_partial_k1_k2}
(the asymptotic expression is given by Eq.~\eqref{eq:pr_GRK} in
Theorem~\ref{theorem:pr_GRK}). The success probability of finding the full target state from the same GRK
sequence \(\pr_{x=t}(k_1,k_2)\) is given by 
\begin{equation}
    \pr_{x=t}(k_1,k_2)
    =
    \left|\langle t|G_n G_m^{k_2}G_n^{k_1}|s_n\rangle\right|^2 .
\end{equation}
The minimal expected iteration number is then given by
\begin{equation}
\label{eq:E_hybrid_min}
    \E^\text{min}_{\text{hybrid},\parallel}
    =
    \min_{k_\text{tot}}
    \left\{\E_{\text{hybrid},\parallel}(k_\text{tot})\right\}.
\end{equation}

Comparing it with the outer parallel search algorithm defined in Definition \ref{def:outer}, we have the following theorem.
\begin{theorem}
\label{thm:hybrid_outer}
    Hybrid parallel search algorithm always outperforms outer parallel search algorithm, namely $\E^\mathrm{min}_{\mathrm{hybrid},\parallel}<\E^\mathrm{min}_{\mathrm{outer},\parallel}$. 
\end{theorem}
\begin{proof}
    We set $k_2=0$, which corresponds to the standard Grover algorithm, as a special case of the GRK sequence. At the condition $k_2=0$, we have
\begin{multline}
\label{eq:E_hybrid_l_2}
    \E_{\text{hybrid},\parallel}(k)
    \\=
    \frac{k}
    {1-\left(1-\pr_{x=t}(k)\right)^l
    +\left(\pr_{x=t_1}(k)-\pr_{x=t}(k)\right)^l},
\end{multline}
where $\pr_{x=t_1}(k)$ is the success probability using Grover's algorithm to find the target block given by Eq. \eqref{eq:pr_grover}; $\pr_{x=t}(k)$ is the success probability of Grover's algorithm given by Eq. \eqref{eq:success_pr_of_Grover}. Comparing it with the expected iteration number of outer parallel search algorithm $\E_{\text{outer},\parallel}(k)$, given by \eqref{eq:E_outer}, we have
    \begin{equation}
        \E_{\text{hybrid},\parallel}(k) < \E_{\text{outer},\parallel}(k),
    \end{equation}
    because the overall success probability of hybrid parallel search algorithm (the denominator of $\E_{\text{hybrid},\parallel}(k)$) is larger than that of the outer parallel search algorithm (the denominator of $\E_{\text{outer},\parallel}(k)$). Therefore, we conclude $\E^\mathrm{min}_{\mathrm{hybrid},\parallel}<\E^\mathrm{min}_{\mathrm{outer},\parallel}$. 
\end{proof}

The superiority of the hybrid parallel search algorithm over the outer parallel scheme stems primarily from its ability to combine the measured block-label qubits from each QPU, thereby increasing the overall probability of identifying the target state. Naturally, as the number of parallel units $l$ increases, the efficiency of the hybrid parallel scheme can be expected to approach that of the outer parallel search scheme. We show the proof in Appendix \ref{sec:comparing_hybrid_outer}.

\begin{table}[t]
\centering

\caption{Comparison between the outer and hybrid parallel schemes for \(l=2\). Here \(m=n/2\), and the hybrid scheme is optimized within the GRK sequence. The percentage reduction is defined as 
\(P=(\E^{\min}_{\mathrm{outer},\parallel}
-\E^{\min}_{\mathrm{hybrid},\parallel})/
\E^{\min}_{\mathrm{outer},\parallel}\times 100\%\).}
\begin{tabular}{c|c|c|c|c}
\hline\hline
~~\(n\)~~ & ~~\(m\)~~ & ~~\(\E^{\min}_{\mathrm{outer},\parallel}\)~~ 
& ~~\(\E^{\min}_{\mathrm{hybrid},\parallel}\)~~ & ~~\(P(\%)\)~~ \\
\hline
4  & 2 & 1.3852  & 1.3642  & 1.5175 \\
6  & 3 & 3.5117  & 3.4792  & 0.9245 \\
8  & 4 & 7.7234  & 7.7158  & 0.0993 \\
10 & 5 & 16.0866 & 16.0816 & 0.0313 \\
\hline\hline
\end{tabular}\label{tab:hybrid_outer_l2}
\end{table}

\begin{table}[t]
\centering

\caption{Comparison between the outer and hybrid parallel schemes for \(l=3\). 
Here \(m=2n/3\), and the hybrid scheme is optimized within the GRK sequence. }
\begin{tabular}{c|c|c|c|c}
\hline\hline
~~\(n\)~~ & ~~\(m\)~~ & ~~\(\E^{\min}_{\mathrm{outer},\parallel}\)~~ 
& ~~\(\E^{\min}_{\mathrm{hybrid},\parallel}\)~~ & ~~\(P(\%)\)~~ \\
\hline
6  & 4  & 2.787207  & 2.769031  & 0.6521 \\
9  & 6  & 9.205030  & 9.199600  & 0.0590 \\
12 & 8  & 27.230060 & 27.227683 & 0.0087 \\
15 & 10 & 78.251802 & 78.250970 & 0.0011 \\
\hline\hline
\end{tabular}\label{tab:hybrid_outer_l3}
\end{table}

%textcolor{red}{To more clearly illustrate the degree of improvement $P=(\mathrm{E}_{\text{outer},\parallel}^{\min}-\mathrm{E}_{\text{hybrid},\parallel}^{\min})/\mathrm{E}_{\text{outer},\parallel}^{\min}$ of the hybrid parallel search algorithm over the outer parallel search algorithm, we list normalized $\mathrm{E}_{\text{outer}}^{\min}$, $\mathrm{E}_{\text{hybrid}}^{\min}$, and the corresponding percentage reduction $P$ for results presented in Fig. \ref{fig:parallel}. It can be seen that the improvement diminishes as $l$ increases, and the hybrid parallel scheme offers a practical advantage only for small-scale parallelism.}

For the special case \(k_2=0\), the hybrid scheme reduces to the Grover-based hybrid parallel scheme. From Eqs.~\eqref{eq:E_outer} and \eqref{eq:E_hybrid_l_2}, the additional success probability in the denominator relative to the outer scheme is given by \((\Pr_{x=t_1}-\Pr_{x=t})^l=\mathcal O(1/N)\). Thus, for fixed \(l\), the hybrid scheme is strictly better than the outer scheme at finite \(N\), but the relative improvement is only \(\mathcal O(1/N)\). Equivalently, the absolute reduction in the expected iteration number is only \(O(N^{-1/2})\), so this advantage is only a subleading improvement. Allowing \(k_2\neq0\) does not change the subleading nature of this improvement: the local GRK operator can only redistribute probability inside the target block, and any positive advantage over the outer scheme still comes from the additional block-concatenation channel. Numerically, we verify the improvement of the hybrid parallel scheme over the outer parallel scheme using \(P=(\mathrm{E}_{\text{outer},\parallel}^{\min}-\mathrm{E}_{\text{hybrid},\parallel}^{\min})/\mathrm{E}_{\text{outer},\parallel}^{\min}\). In Tables~\ref{tab:hybrid_outer_l2} and \ref{tab:hybrid_outer_l3}, we see that \(P\) quickly falls below \(0.1\%\), showing that the improvement is subleading as \(N\) increases.

We have also compared the hybrid and inner parallel schemes. Analytically, we find that, for \(l=2\), the inner parallel scheme outperforms the hybrid parallel scheme; see Appendix~\ref{sec:comparing_hybrid_inner} for
the proof. We also expect the inner parallel scheme to outperform the hybrid parallel scheme for larger admissible values of \(l\), because the hybrid parallel scheme approaches the outer parallel scheme as \(l\) increases, while the outer parallel scheme is itself less efficient than the inner parallel
scheme.

%In Theorems \ref{thm:hybrid_outer} and \ref{thm:hybrid_inner}, the theoretical analysis assumes $k_2=0$, which corresponds to a special case of the GRK sequence. We perform a numerical search for the optimal sequences minimizing $\E^\text{min}_{\text{hybrid},\parallel}$ defined in Eq. \eqref{eq:E_hybrid_min}, scanning over $n=10$ to $n=50$ and all corresponding $m<n$. The results show that only for a few instances with $k_2\neq 0$ provide an additional advantage. However, because $l$ does not divide $n$ in those cases, they do not provide a practical advantage. 

\section{\label{sec:conclusion} Conclusions}

In this work, we investigated whether the GRK operator sequence, namely $G_n G_m^{k_2} G_n^{k_1}$, remains optimal under a fixed total number of oracle queries $k_{\mathrm{tot}}$. Through an exhaustive numerical enumeration of all $2^{k_{\mathrm{tot}}}$ admissible operator compositions, we provided strong evidence that the GRK structure maximizes the success probability for each fixed $k_{\mathrm{tot}}$. Under this numerically supported GRK ansatz, we derived asymptotically tight bounds on the maximal success probability and the minimal expected iteration number for the quantum partial search algorithm, as stated in Theorems~\ref{theorem:pr_GRK} and~\ref{theorem:k_opt_GRK_partial_search}. In particular, our analysis shows that the optimal choice of $k_2$ is independent of both the database size and $k_{\mathrm{tot}}$. Building on the established efficiency limits of partial search, we analyzed a GRK-based parallel search scheme and showed that it does not outperform existing parallel Grover strategies. We then introduced a hybrid parallel protocol, which provides a strict but subleading improvement over the outer parallel scheme, while the inner parallel scheme remains superior whenever it
is applicable.

We conclude by outlining three open questions concerning the quantum partial search algorithm. First, our analysis has focused on databases containing a single marked item. Although the quantum partial search algorithm admits extensions to the multi-target setting \cite{Choi2007QuantumPS,zhong2009quantum,Zhang2017QuantumPS}, its efficiency in that case depends sensitively on the distribution of target states across blocks. A systematic study of the hybrid parallel search scheme in multi-target scenarios therefore remains an important direction for future work.

Second, Theorem~\ref{theorem:k_opt_GRK_partial_search} indicates that when $m > \left\lfloor n/2 \right\rfloor$ (equivalently, when the number of blocks is less than $2^{n/2}$), the GRK algorithm becomes inefficient relative to classical search. This observation raises the broader question of whether alternative formulations of quantum partial search could retain a quantum advantage even in the regime $m > \left\lfloor n/2 \right\rfloor$.

Third, although we focus on the optimality of the GRK sequence with respect to oracle-query complexity, this does not imply optimality in circuit depth. Previous studies have emphasized that query complexity and depth complexity constitute distinct resource measures in quantum search algorithms \cite{grover2002trade,Zhang2020DepthO,brianski2021introducing}. From a practical implementation perspective, the design of quantum partial search algorithms that are optimal in circuit depth remains a meaningful and challenging objective.

\begin{acknowledgments}
This work was supported by the NSFC (Grants No.12305028, No.12275215, and No.12247103), and the Youth Innovation Team of Shaanxi Universities. KZ is supported by the China Postdoctoral Science Foundation under Grant Number 2025M773421, Shaanxi Province Postdoctoral Science Foundation under Grant Number 2025BSHYDZZ017, and Scientific Research Program Funded by Education Department of Shaanxi Provincial Government (Program No.24JP186). VK is funded by the U.S. Department of Energy, Office of Science, National Quantum Information Science Research Centers, Co-Design Center for Quantum Advantage (C2QA) under Contract No. DE-SC0012704.
\end{acknowledgments}

\appendix

\section{Numerical investigations on the optimal quantum partial search algorithm}

\label{app:numerical_quest}

In Sec. \ref{sec:num_quest}, we have examined all $2^{k_\text{tot}}$ possible operator sequences for the quantum partial search algorithm, as defined in Eq. \eqref{eq:S_sequence}. The optimal sequences and their corresponding maximal success probabilities are summarized in Table \ref{tab:n=8,pr} for $n=8$. With the exception of a few entries in the lower-right corner of the table, all optimal sequences found in our enumeration adopt the GRK form, namely $G_nG_m^{k_2}G_n^{k_1}$. These exceptional cases occur in the regime where the GRK sequence has already passed its natural optimal amplification point. In this situation, the remaining oracle queries can be used redundantly to further adjust the state and slightly enhance the success probability, in a manner analogous to the redundant-query strategy used in the full-search algorithm~\cite{wang2025near}. Therefore, these exceptional entries should not be interpreted as evidence that the GRK sequence is intrinsically suboptimal. Rather, they reflect finite-query over-rotation effects after the main GRK-type amplitude amplification has already been achieved.

The corresponding expected iteration numbers are presented in Table \ref{tab:n=8,k}. Notably, for $m=6$ and $m=7$, the minimal expected iteration numbers are achieved by sequences with $k_\text{tot}=2$, meaning that the quantum partial search algorithm operates inefficiently. This result underscores the algorithm's diminishing performance for large values of $m$.  We have also tested the optimal operator sequences for larger database sizes with a limited number of iterations in Tables \ref{tab:k20_15_20_25} and \ref{tab:k20_40}. We still observe that the GRK form is optimal in all cases considered, which suggests that the GRK form remains optimal in a broad range of practical cases.

\begin{table*}[t]
\caption{Optimal operators and their success probabilities for the quantum partial search algorithm with $n=8$, corresponding to Fig. \ref{fig:pr_max}. Algorithm sequences highlighted in blue deviate from the standard GRK form.}
\centering
\begin{tabular}{c|c|c|c|c|c|c|c}
\hline\hline
\multicolumn{2}{c|}{$m$} & 2 & 3 & 4 & 5 & 6 & 7 \\\hline 
\multirow{2}{*}{$k_\text{tot} = 2$} & Operator & $G_8G_2$ & $G_8G_3$ & $G_8G_4$ & $G_8G_5$ & $G_8G_6$ & $G_8G_7$ \\
& $\pr_{x=t_1}$ (\%) & 10.5747 & 13.4310 & 16.9145 & 22.7749 & 33.9446 & 56.0089 \\\hline 
\multirow{2}{*}{$k_\text{tot} = 3$} & Operator & $G_8G_2G_8$ & $G_8G_3G_8$ & $G_8G^2_4$ & $G_8G^2_5$ & $G_8G^2_6$ & $G_8G^2_7$ \\
& $\pr_{x=t_1}$ (\%) & 19.0236 & 23.0100 & 27.0968 & 33.6418 & 43.8606 & 62.8228 \\\hline 
\multirow{2}{*}{$k_\text{tot} = 4$} & Operator & $G_8G_2G^2_8$ & $G_8G_3G^2_8$ & $G_8G^2_4G_8$ & $G_8G^3_5$ & $G_8G^3_6$ & $G_8G^3_7$ \\
& $\pr_{x=t_1}$ (\%) & 29.4353 & 34.2693 & 38.7202 & 45.7859 & 55.6425 & 71.1470 \\\hline 
\multirow{2}{*}{$k_\text{tot} = 5$} & Operator & $G_8G_2G^3_8$ & $G_8G_3G^3_8$ & $G_8G_4^2G^2_8$ & $G_8G_5^3G_8$ & $G_8G_6^4$ & $G_8G_7^4$ \\
& $\pr_{x=t_1}$ (\%) & 41.1617 & 46.5081 & 51.0612 & 58.2046 & 67.7145 & 79.9491 \\\hline 
\multirow{2}{*}{$k_\text{tot} = 6$} & Operator & $G_8G_2G^4_8$ & $G_8G_3G^4_8$ & $G_8G_4^2G^3_8$ & $G_8G_5^3G^2_8$ & $G_8G_6^4G_8$ & $G_8G_7^5$ \\
& $\pr_{x=t_1}$ (\%) & 53.4727 & 58.9644 & 63.3514 & 70.1245 & 78.7004 & 88.1374 \\\hline 
\multirow{2}{*}{$k_\text{tot} = 7$} & Operator & $G_8G_2G^5_8$ & $G_8G_3G^5_8$ & $G_8G^2_4G^4_8$ & $G_8G^3_5G^3_8$ & $G_8G^4_6G^2_8$ & $G_8G_7^6$ \\
& $\pr_{x=t_1}$ (\%) & 65.6019 & 70.8626 & 74.8257 & 80.8038 & 87.9164 & 94.6963 \\\hline 
\multirow{2}{*}{$k_\text{tot} = 8$} & Operator & $G_8G_2G^6_8$ & $G_8G_3G^6_8$ & $G_8G^2_4G^5_8$ & $G_8G^3_5G^4_8$ & $G_8G^4_6G^3_8$ & $G_8G^7_7$ \\
& $\pr_{x=t_1}$ (\%) & 76.7941 & 81.4622 & 84.7698 & 89.5775 & 94.7887 & 98.8124 \\\hline 
\multirow{2}{*}{$k_\text{tot} = 9$} & Operator & $G_8G_2G^7_8$ & $G_8G_3G^7_8$ & $G_8G_4^2G^6_8$ & $G_8G_5^3G^5_8$ & $G_8G_6^4G^4_8$ & \textcolor{blue}{$G_8G^6_7G_8G_7$} \\
& $\pr_{x=t_1}$ (\%) & 86.3525 & 90.1031 & 92.5645 & 95.8994 & 98.8894 & \textcolor{blue}{99.9998} \\\hline 
\multirow{2}{*}{$k_\text{tot} = 10$} & Operator & $G_8G_2G^8_8$ & $G_8G_3G^8_8$ & $G_8G_4^2G^7_8$ & $G_8G_5^3G^6_8$ & \textcolor{blue}{$G_8G_6^4G^3_8G_6^2$} & \textcolor{blue}{$G_8G_7^3(G_8G_7)^2$} \\
& $\pr_{x=t_1}$ (\%) & 93.6822 & 96.2474 & 97.7247 & 99.3760 & \textcolor{blue}{99.9999} & \textcolor{blue}{99.9999} \\\hline 
\multirow{2}{*}{$k_\text{tot} = 11$} & Operator & $G_8G_2G^9_8$ & $G_8G_3G^9_8$ & $G_8G_4^2G^8_8$ & \textcolor{blue}{$G_8G_5(G^2_8G_5)^2G_8G^2_5$} & \textcolor{blue}{$G_8G_6G_8^4(G_6G_8)^2G_8$} & \textcolor{blue}{$G_7G_8G_7^3(G_8G_7)^3$} \\
& $\pr_{x=t_1}$ (\%) & 98.3268 & 99.5126 & 99.9290 & \textcolor{blue}{99.9999} & \textcolor{blue}{99.9999} & \textcolor{blue}{99.9999} \\
\hline\hline
\end{tabular}
\label{tab:n=8,pr}
\end{table*}

\begin{table*}[t]
\caption{Optimal operator sequences and corresponding expected iteration counts for the quantum partial search algorithm with $n=8$, corresponding to Fig. \ref{fig:E}. Values and operators highlighted in red indicate, for each $m$, the sequence achieving the minimal expected iteration number (i.e., the smallest value in each column). Sequences marked in blue deviate from the standard GRK form. These exceptional cases occur only when a shorter GRK‑type sequence already drives the success probability extremely close to unity; the additional queries never reduce the expected number of iterations and therefore represent genuinely sub‑optimal, redundant extensions.}
\centering
\begin{tabular}{c|c|c|c|c|c|c|c}
\hline\hline
\multicolumn{2}{c|}{$m$} & 2 & 3 & 4 & 5 & 6 & 7 \\\hline 
\multirow{2}{*}{$k_\text{tot} = 2$} & Operator & $G_8G_2$ & $G_8G_3$ & $G_8G_4$ & $G_8G_5$ & \textcolor{red}{$G_8G_6$} & \textcolor{red}{$G_8G_7$} \\
& $\E_{x=t_1}$ & 18.9130 & 14.8909 & 11.8242 & 8.7816 & \textcolor{red}{5.8919} & \textcolor{red}{3.5709} \\\hline 
\multirow{2}{*}{$k_\text{tot} = 3$} & Operator & $G_8G_2G_8$ & $G_8G_3G_8$ & $G_8G^2_4$ & $G_8G^2_5$ & $G_8G^2_6$ & $G_8G^2_7$ \\
& $\E_{x=t_1}$ & 15.7699 & 13.0378 & 11.0714 & 8.9175 & 6.8398 & 4.7753 \\\hline 
\multirow{2}{*}{$k_\text{tot} = 4$} & Operator & $G_8G_2G^2_8$ & $G_8G_3G^2_8$ & $G_8G^2_4G_8$ & $G_8G^3_5$ & $G_8G^3_6$ & $G_8G^3_7$ \\
& $\E_{x=t_1}$ & 13.5891 & 11.6722 & 10.3305 & 8.7363 & 7.1887 & 5.6222 \\\hline 
\multirow{2}{*}{$k_\text{tot} = 5$} & Operator & $G_8G_2G^3_8$ & $G_8G_3G^3_8$ & $G_8G_4^2G^2_8$ & $G_8G_5^3G_8$ & $G_8G_6^4$ & $G_8G_7^4$ \\
& $\E_{x=t_1}$ & 12.1472 & 10.7508 & 9.7922 & 8.5904 & 7.3839 & 6.2540 \\\hline 
\multirow{2}{*}{$k_\text{tot} = 6$} & Operator & $G_8G_2G^4_8$ & $G_8G_3G^4_8$ & $G_8G_4^2G^3_8$ & \textcolor{red}{$G_8G_5^3G^2_8$} & $G_8G_6^4G_8$ & $G_8G_7^5$ \\
& $\E_{x=t_1}$ & 11.2207 & 10.1756 & 9.4710 & \textcolor{red}{8.5562} & 7.6238 & 6.8076 \\\hline 
\multirow{2}{*}{$k_\text{tot} = 7$} & Operator & $G_8G_2G^5_8$ & $G_8G_3G^5_8$ & $\textcolor{red}{G_8G^2_4G^4_8}$ & $G_8G^3_5G^3_8$ & $G_8G^4_6G^2_8$ & $G_8G_7^6$ \\
& $\E_{x=t_1}$ & 10.6704 & 9.8783 & \textcolor{red}{9.3551} & 8.6630 & 7.9621 & 7.3921 \\\hline 
\multirow{2}{*}{$k_\text{tot} = 8$} & Operator & \textcolor{red}{$G_8G_2G^6_8$} & \textcolor{red}{$G_8G_3G^6_8$} & $G_8G^2_4G^5_8$ & $G_8G^3_5G^4_8$ & $G_8G^4_6G^3_8$ & $G_8G^7_7$ \\
& $\E_{x=t_1}$ & \textcolor{red}{10.4175} & \textcolor{red}{9.8205} & 9.4373 & 8.9308 & 8.4398 & 8.0962 \\\hline 
\multirow{2}{*}{$k_\text{tot} = 9$} & Operator & $G_8G_2G^7_8$ & $G_8G_3G^7_8$ & $G_8G_4^2G^6_8$ & $G_8G_5^3G^5_8$ & $G_8G_6^4G^4_8$ & \textcolor{blue}{$G_8G_7^6G_8G_7$} \\
& $\E_{x=t_1}$ & 10.4224 & 9.9886 & 9.7229 & 9.3848 & 9.1011 & \textcolor{blue}{9.0000} \\\hline 
\multirow{2}{*}{$k_\text{tot} = 10$} & Operator & $G_8G_2G^8_8$ & $G_8G_3G^8_8$ & $G_8G_4^2G^7_8$ & $G_8G_5^3G^6_8$ & \textcolor{blue}{$G_8G_6^4G^3_8G_6^2$} & \textcolor{blue}{$G_8G_7^3(G_8G_7)^2$} \\
& $\E_{x=t_1}$ & 10.6744 & 10.3899 & 10.2328 & 10.0628 & \textcolor{blue}{10.0000} & \textcolor{blue}{10.0000} \\\hline 
\multirow{2}{*}{$k_\text{tot} = 11$} & Operator & $G_8G_2G^9_8$ & $G_8G_3G^9_8$ & $G_8G_4^2G^8_8$ & \textcolor{blue}{$G_8G_5(G^2_8G_5)^2G_8G^2_5$} & \textcolor{blue}{$G_8G_6G_8^4(G_6G_8)^2G_8$} & \textcolor{blue}{$G_7G_8G_7^3(G_8G_7)^3$} \\
& $\E_{x=t_1}$ & 11.1872 & 11.0539 & 11.0078 & \textcolor{blue}{11.0000} & \textcolor{blue}{11.0000} & \textcolor{blue}{11.0000} \\
\hline\hline
\end{tabular}
\label{tab:n=8,k}
\end{table*}

\begin{table*}[t]
\caption{Optimal operator sequences and corresponding expected iteration counts for the quantum partial search algorithm with $k_\text{tot}=20$. The database size is chosen as $n=15,20,25$. The size of partial search is chosen as $2\leq m\leq n-1$.}
\centering
\begin{tabular}{c|c|c|c|c|c|c|c|c}
\hline\hline
\multicolumn{9}{c}{$n=15$}\\\hline
$m$ & 2 & 3 & 4 & 5 & 6 & 7 & 8 & 9 \\
Operator & $G_{15}G_2G_{15}^{18}$ & $G_{15}G_3G_{15}^{18}$ & $G_{15}G_4^2G_{15}^{17}$ & $G_{15}G_5^2G_{15}^{17}$ & $G_{15}G_6^4G_{15}^{15}$ & $G_{15}G_7^5G_{15}^{14}$ & $G_{15}G_8^8G_{15}^{11}$ & $G_{15}G_9^{11}G_{15}^8$ \\
$\E_{x=t_1}$ & 395.9080 & 378.2546 & 367.0226 & 348.2238 & 323.7248 & 294.4253 & 259.0569 & 219.5925 \\\hline
$m$ & 10 & 11 & 12 & 13 & 14 & & & \\
Operator & $G_{15}G_{10}^{17}G_{15}^2$ & $G_{15}G_{11}^{19}$ & $G_{15}G_{12}^{19}$ & $G_{15}G_{13}^{19}$ & $G_{15}G_{14}^{19}$ & & & \\
$\E_{x=t_1}$ & 177.7726 & 136.7080 & 97.2027 & 62.4357 & 36.5457 & & & \\
\hline\hline
\multicolumn{9}{c}{$n=20$}\\\hline
$m$ & 2 & 3 & 4 & 5 & 6 & 7 & 8 & 9 \\
Operator & $G_{20}G_2G_{20}^{18}$ & $G_{20}G_3G_{20}^{18}$ & $G_{20}G_4^2G_{20}^{17}$ & $G_{20}G_5^2G_{20}^{17}$ & $G_{20}G_6^4G_{20}^{15}$ & $G_{20}G_7^5G_{20}^{14}$ & $G_{20}G_8^8G_{20}^{11}$ & $G_{20}G_9^{11}G_{20}^8$ \\
$\E_{x=t_1}$ & $1.2460\times 10^{4}$ & $1.1895\times 10^{4}$ & $1.1536\times 10^{4}$ & $1.0934\times 10^{4}$ & $1.0150\times 10^{4}$ & $9.2112\times 10^{3}$ & $8.0800\times 10^{3}$ & $6.8153\times 10^{3}$ \\\hline
$m$ & 10 & 11 & 12 & 13 & 14 & 15 & 16 & 17 \\
Operator & $G_{20}G_{10}^{16}G_{20}^3$ & $G_{20}G_{11}^{19}$ & $G_{20}G_{12}^{19}$ & $G_{20}G_{13}^{19}$ & $G_{20}G_{14}^{19}$ & $G_{20}G_{15}^{19}$ & $G_{20}G_{16}^{19}$ & $G_{20}G_{17}^{19}$ \\
$\E_{x=t_1}$ & $5.4764\times 10^{3}$ & $4.1631\times 10^{3}$ & $2.9212\times 10^{3}$ & $1.8537\times 10^{3}$ & $1.0751\times 10^{3}$ & 584.9292 & 306.0301 & 156.6593 \\\hline
$m$ & 18 & 19 & & & & & & \\
Operator & $G_{20}G_{18}^{19}$ & $G_{20}G_{19}^{19}$ & & & & & & \\
$\E_{x=t_1}$ & 79.2751 & 39.8784 & & & & & & \\
\hline\hline
\multicolumn{9}{c}{$n=25$}\\\hline
$m$ & 2 & 3 & 4 & 5 & 6 & 7 & 8 & 9 \\
Operator & $G_{25}G_2G_{25}^{18}$ & $G_{25}G_3G_{25}^{18}$ & $G_{25}G_4^2G_{25}^{17}$ & $G_{25}G_5^2G_{25}^{17}$ & $G_{25}G_6^4G_{25}^{15}$ & $G_{25}G_7^5G_{25}^{14}$ & $G_{25}G_8^8G_{25}^{11}$ & $G_{25}G_9^{11}G_{25}^8$ \\
$\E_{x=t_1}$ & $3.9852\times 10^{5}$ & $3.8044\times 10^{5}$ & $3.6894\times 10^{5}$ & $3.4967\times 10^{5}$ & $3.2460\times 10^{5}$ & $2.9455\times 10^{5}$ & $2.5835\times 10^{5}$ & $2.1788\times 10^{5}$ \\\hline
$m$ & 10 & 11 & 12 & 13 & 14 & 15 & 16 & 17 \\
Operator & $G_{25}G_{10}^{16}G_{25}^3$ & $G_{25}G_{11}^{19}$ & $G_{25}G_{12}^{19}$ & $G_{25}G_{13}^{19}$ & $G_{25}G_{14}^{19}$ & $G_{25}G_{15}^{19}$ & $G_{25}G_{16}^{19}$ & $G_{25}G_{17}^{19}$ \\
$\E_{x=t_1}$ & $1.7504\times 10^{5}$ & $1.3301\times 10^{5}$ & $9.3298\times 10^{4}$ & $5.9181\times 10^{4}$ & $3.4315\times 10^{4}$ & $1.8666\times 10^{4}$ & $9.7651\times 10^{3}$ & $4.9986\times 10^{3}$ \\\hline
$m$ & 18 & 19 & 20 & 21 & 22 & 23 & 24 & \\
Operator & $G_{25}G_{18}^{19}$ & $G_{25}G_{19}^{19}$ & $G_{25}G_{20}^{19}$ & $G_{25}G_{21}^{19}$ & $G_{25}G_{22}^{19}$ & $G_{25}G_{23}^{19}$ & $G_{25}G_{24}^{19}$ & \\
$\E_{x=t_1}$ & $2.5294\times 10^{3}$ & $1.2724\times 10^{3}$ & 638.1138 & 319.5429 & 159.8933 & 79.9771 & 39.9962 & \\
\hline\hline
\end{tabular}
\label{tab:k20_15_20_25}
\end{table*}

\begin{table*}[t]
\caption{Optimal operator sequences and corresponding expected iteration counts for the quantum partial search algorithm with $k_\text{tot}=20$ (II: $n=40$)}
\centering
\begin{tabular}{c|c|c|c|c|c|c}
\hline\hline
\multicolumn{7}{c}{$n=40$}\\\hline
$m$ & 2 & 3 & 4 & 5 & 6 & 7 \\
Operator & $G_{40}^{20}$ & $G_{40}G_3G_{40}^{18}$ & $G_{40}G_4^2G_{40}^{17}$ & $G_{40}G_5^2G_{40}^{17}$ & $G_{40}G_6^4G_{40}^{15}$ & $G_{40}G_7^5G_{40}^{14}$ \\
$\E_{x=t_1}$ & $1.3058\times 10^{10}$ & $1.2466\times 10^{10}$ & $1.2089\times 10^{10}$ & $1.1458\times 10^{10}$ & $1.0636\times 10^{10}$ & $9.6517\times 10^{9}$ \\\hline
$m$ & 8 & 9 & 10 & 11 & 12 & 13 \\
Operator & $G_{40}G_8^8G_{40}^{11}$ & $G_{40}G_9^{11}G_{40}^8$ & $G_{40}G_{10}^{16}G_{40}^3$ & $G_{40}G_{11}^{19}$ & $G_{40}G_{12}^{19}$ & $G_{40}G_{13}^{19}$ \\
$\E_{x=t_1}$ & $8.4655\times 10^{9}$ & $7.1394\times 10^{9}$ & $5.7354\times 10^{9}$ & $4.3584\times 10^{9}$ & $3.0570\times 10^{9}$ & $1.9391\times 10^{9}$ \\\hline
$m$ & 14 & 15 & 16 & 17 & 18 & 19 \\
Operator & $G_{40}G_{14}^{19}$ & $G_{40}G_{15}^{19}$ & $G_{40}G_{16}^{19}$ & $G_{40}G_{17}^{19}$ & $G_{40}G_{18}^{19}$ & $G_{40}G_{19}^{19}$ \\
$\E_{x=t_1}$ & $1.1244\times 10^{9}$ & $6.1160\times 10^{8}$ & $3.1995\times 10^{8}$ & $1.6378\times 10^{8}$ & $8.2875\times 10^{7}$ & $4.1689\times 10^{7}$ \\\hline
$m$ & 20 & 21 & 22 & 23 & 24 & 25 \\
Operator & $G_{40}G_{20}^{19}$ & $G_{40}G_{21}^{19}$ & $G_{40}G_{22}^{19}$ & $G_{40}G_{23}^{19}$ & $G_{40}G_{24}^{19}$ & $G_{40}G_{25}^{19}$ \\
$\E_{x=t_1}$ & $2.0908\times 10^{7}$ & $1.0470\times 10^{7}$ & $5.2389\times 10^{6}$ & $2.6204\times 10^{6}$ & $1.3105\times 10^{6}$ & $6.5530\times 10^{5}$ \\\hline
$m$ & 26 & 27 & 28 & 29 & 30 & 31 \\
Operator & $G_{40}G_{26}^{19}$ & $G_{40}G_{27}^{19}$ & $G_{40}G_{28}^{19}$ & $G_{40}G_{29}^{19}$ & $G_{40}G_{30}^{19}$ & $G_{40}G_{31}^{19}$ \\
$\E_{x=t_1}$ & $3.2766\times 10^{5}$ & $1.6384\times 10^{5}$ & $8.1919\times 10^{4}$ & $4.0960\times 10^{4}$ & $2.0480\times 10^{4}$ & $1.0240\times 10^{4}$ \\\hline
$m$ & 32 & 33 & 34 & 35 & 36 & 37 \\
Operator & $G_{40}G_{32}^{19}$ & $G_{40}G_{33}^{19}$ & $G_{40}G_{34}^{19}$ & $G_{40}G_{35}^{19}$ & $G_{40}G_{36}^{19}$ & $G_{40}G_{37}^{19}$ \\
$\E_{x=t_1}$ & $5.1200\times 10^{3}$ & $2.5600\times 10^{3}$ & $1.2800\times 10^{3}$ & 639.9999 & 320.0000 & 160.0000 \\\hline
$m$ & 38 & 39 & & & & \\
Operator & $G_{40}G_{38}^{19}$ & $G_{40}G_{39}^{19}$ & & & & \\
$\E_{x=t_1}$ & 80.0000 & 40.0000 & & & & \\
\hline\hline
\end{tabular}
\label{tab:k20_40}
\end{table*}

\section{Comparing the hybrid and outer parallel schemes}

\label{sec:comparing_hybrid_outer}

\begin{theorem}
\label{thm:parallel_hybrid_outer_large_l}
When the parallelism reaches its maximum, i.e., when $l=n$, a comparison between the hybrid and outer parallel schemes yields
\begin{equation}
\E_{\mathrm{hybrid},\parallel}^\mathrm{min}(l=n) \lessapprox \E_{\mathrm{outer},\parallel}^\mathrm{min}(l=n).
\end{equation}
\end{theorem}
\begin{proof}

The condition \(l=n\) gives \(m=n-1\) and hence \(\sin\gamma=1/\sqrt{2}\).
We consider the special case \(k_2=0\), which corresponds to the standard
Grover limit of the GRK sequence. This is sufficient for obtaining an upper
bound on the minimal expected iteration number of the hybrid scheme. In this
case, for \(n\gg 1\), we have
\begin{subequations}
\begin{gather}
    \pr_{x=t_1}(k)
    \approx
    \sin^2\phi+\frac{1}{2}\cos^2\phi, \\
    \pr_{x=t}(k)
    \approx
    \sin^2\phi,
\end{gather}
\end{subequations}
where \(\phi=2k\theta_1\). Therefore,
\begin{equation}
    \pr_{x=t_1}(k)-\pr_{x=t}(k)
    \approx
    \frac{1}{2}\cos^2\phi .
\end{equation}
The success probability of the hybrid parallel scheme, given by Eq. \eqref{eq:E_hybrid_l_2}, gives
\begin{multline}
    P_{\mathrm{hybrid},\parallel}(k)
    \\=
    1-\left(1-\pr_{x=t}(k)\right)^l
    +
    \left(\pr_{x=t_1}(k)-\pr_{x=t}(k)\right)^l ,
\end{multline}
and setting \(l=n\), we obtain
\begin{align}
    P_{\mathrm{hybrid},\parallel}(\phi)
    &\approx
    1-\left(1-\sin^2\phi\right)^n
    +
    \left(\frac{1}{2}\cos^2\phi\right)^n  \nonumber \\
    &=
    1-\cos^{2n}\phi+2^{-n}\cos^{2n}\phi  \nonumber \\
    &=
    1-\left(1-2^{-n}\right)\cos^{2n}\phi .
\end{align}
Since \(k\approx \phi/(2\theta_1)\approx 2^{n/2-1}\phi\), the expected
iteration number with \(k_2=0\) is
\begin{equation}
    \E_{\mathrm{hybrid},\parallel}(\phi,l=n)
    \approx
    \frac{2^{n/2-1}\phi}
    {1-\cos^{2n}\phi+2^{-n}\cos^{2n}\phi}.
\end{equation}
For \(n\gg 1\), we have \(\phi=\mathcal{O}(1/\sqrt{n})\).
Using
\begin{equation}
    \cos^{2n}\phi
    =
    \exp\!\left(2n\ln\cos\phi\right)
    =
    \exp\!\left(-n\phi^2+\mathcal{O}(n\phi^4)\right),
\end{equation}
and noting that \(n\phi^4=\mathcal{O}(1/n)\) in this regime, we have
\begin{equation}
    \cos^{2n}\phi
    \approx
    e^{-n\phi^2}.
\end{equation}
The additional hybrid correction term satisfies
\begin{equation}
    2^{-n}\cos^{2n}\phi \leq 2^{-n},
\end{equation}
and is therefore exponentially suppressed as \(n\to\infty\). Hence the
denominator reduces to
\begin{equation}
    1-\cos^{2n}\phi+2^{-n}\cos^{2n}\phi
    \approx
    1-e^{-n\phi^2}.
\end{equation}
Thus we have
\begin{equation}
    \E_{\mathrm{hybrid},\parallel}(\phi,l=n)
    \approx
    \frac{2^{n/2-1}\phi}{1-e^{-n\phi^2}}.
\end{equation}
The minimal point is given by \(\phi^*\) satisfying 
\((1+2n\phi_\text{min}^2)\exp\left(-n\phi_\text{min}^2\right)=1\), which has the numerical solution \(n\phi_\text{min}^2\approx 1.25643\). Thus \(k_{\min}\approx0.56045\sqrt{N/n}\)
and the corresponding minimum expected iteration number is
\begin{equation}
\E_{\text{hybrid},\parallel}^{\min}\approx0.7835\sqrt{\frac{N}{n}},
\end{equation}
same value as \(\E_{\text{outer},\parallel}^{\min}\) given by Eq. \eqref{eq:E_outer_parallel}.
\end{proof}

\section{Comparing the hybrid and inner parallel schemes}

\label{sec:comparing_hybrid_inner}

\begin{theorem}
\label{thm:hybrid_inner}
Hybrid parallel search algorithm does not outperform inner parallel search algorithm when $l=2$, namely $\E^\mathrm{min}_{\mathrm{hybrid},\parallel}(l=2)>\E^\mathrm{min}_{\mathrm{inner},\parallel}(l=2)$. 
\end{theorem}
\begin{proof}
We first consider the special case \(k_2=0\), which corresponds to the standard Grover limit of the GRK sequence. For \(l=2\), we have \(m=n/2\) and hence \(b=\sqrt{N}\). In this case,
\begin{equation}
    \pr_{x=t}(k)
    =
    \sin^2\phi ,
\end{equation}
and
\begin{equation}
    \pr_{x=t_1}(k)
    =
    \sin^2\phi
    +
    \frac{\sqrt{N}-1}{N-1}\cos^2\phi ,
\end{equation}
where \(\phi=(2k+1)\theta_1\simeq 2k\theta_1\) in the large-\(N\) limit.
Thus
\begin{equation}
    \pr_{x=t_1}(k)-\pr_{x=t}(k)
    =
    \frac{\sqrt{N}-1}{N-1}\cos^2\phi
    =
    \mathcal{O}\!\left(N^{-1/2}\right).
\end{equation}

The hybrid success probability given by Eq. \eqref{eq:E_hybrid_l_2} becomes
\begin{align}
    P_{\mathrm{hybrid},\parallel}(k)
    &=
    1-\cos^4\phi
    +
    \left(
    \frac{\sqrt{N}-1}{N-1}
    \right)^2
    \cos^4\phi  \nonumber \\
    &=
    1-\cos^4\phi+\mathcal{O}\!\left(N^{-1}\right).
\end{align}
Therefore the expected iteration number with \(k_2=0\) is
\begin{equation}
    \E_{\mathrm{hybrid},\parallel}(\phi,l=2)
    =
    \frac{k}{P_{\mathrm{hybrid},\parallel}(k)}
    \approx
    \frac{\phi\sqrt{N}}
    {2\left(1-\cos^4\phi\right)} .
\end{equation}
The minimum is given by $\phi_* \approx 0.80718$, which gives
\begin{equation}
    \E_{\mathrm{hybrid},\parallel}^{(k_2=0),\mathrm{min}}(l=2)
    \approx
    \frac{\phi_*}
    {2\left(1-\cos^4\phi_*\right)}
    \sqrt{N}
    \approx
    0.5233\sqrt{N}.
\end{equation}

Next we exam the cases with \(k_2\neq0\). The leading local rotation angle is denoted as \(\beta=k_2\theta_2\). To leading order in the large-\(N\) limit with
\(m=n/2\), the target-block success probability and the full-target success
probability take the form
\begin{subequations}
\begin{align}
    &\pr_{x=t_1}(k_1,k_2)
    =
    \sin^2\phi+\mathcal{O}\!\left(N^{-1/4}\right), \\
    &\pr_{x=t}(k_1,k_2)
    =
    \sin^2\phi\cos^2(2\beta)
    +\mathcal{O}\!\left(N^{-1/4}\right).
\end{align}
\end{subequations}
Hence the hybrid success probability satisfies
\begin{align}
P_{\mathrm{hybrid},\parallel}
&=
\sin^4\phi
+
2\sin^2\phi\cos^2\phi\cos^2(2\beta)
+O\!\left(N^{-1/4}\right)  \nonumber\\
&\le
1-\cos^4\phi
+O\!\left(N^{-1/4}\right).
\end{align}
Thus, at leading order, the success probability is maximized by
\(\cos^2(2\beta)=1\), which corresponds to the \(k_2=0\) Grover branch up to
equivalent local rotations. Nonzero local GRK steps cannot improve the leading
coefficient \(0.5233\).

On the other hand, the minimal expected iteration number of the inner parallel
scheme for \(l=2\) is
\begin{equation}
    \E^\mathrm{min}_{\mathrm{inner},\parallel}(l=2)
    \approx
    0.69\sqrt{\frac{N}{2}}
    \approx
    0.4879\sqrt{N}.
\end{equation}
Therefore, we establish $\E^\mathrm{min}_{\mathrm{hybrid},\parallel}(l=2)>\E^\mathrm{min}_{\mathrm{inner},\parallel}(l=2)$, namely the inner parallel scheme outperforms the hybrid parallel scheme at $l=2$. 
\end{proof}

%\bibliographystyle{apsrev4-2}
%\bibliography{bib}

\providecommand{\noopsort}[1]{}

\end{document}